\documentclass{vldb}

\usepackage[utf8]{inputenc}
\usepackage{balance}  
\usepackage[labelfont=bf,textfont=md]{caption}
\usepackage{graphicx}
\usepackage{algorithm}
\usepackage[noend]{algorithmic}
\usepackage{amsmath}
\usepackage{dsfont}
\usepackage[dvipsnames]{xcolor}
\usepackage{subfig}
\usepackage{hyperref}
\usepackage[normalem]{ulem}
\usepackage{xspace}
\usepackage{color}
\usepackage{colortbl}

\newtheorem{theorem}{Theorem}[section]
\newtheorem{example}{Example}[section]

\def\oriDistr{P}      
\def\propDistr{Q}	

\def\qOneIntro{Q_1}
\def\qTwoIntro{Q_2}
\def\countq{\mathtt{count}}  
\def\topk{\mathtt{top}}  

\def\numPatterns{z}  
\def\numSubRankings{w}  
\def\numPropDistrs{d}

\def\embedding{\delta}
\def\ungroundedVariables{V^+}
\def\partialOrder{\boldsymbol{\upsilon}}
\def\subRanking{\boldsymbol{\psi}}
\def\extension{\Delta}
\def\compensation{c}
\def\benchmarkA{\textbf{Benchmark-A}}
\def\benchmarkB{\textbf{Benchmark-B}}
\def\benchmarkC{\textbf{Benchmark-C}}
\def\benchmarkD{\textbf{Benchmark-D}}

\def\e#1{{\em #1}}
\def\val#1{\texttt{#1}}
\def\angs#1{\mathord{\langle #1 \rangle}}
\def\btau{\boldsymbol{\tau}} 
\def\bsigma{\boldsymbol{\sigma}} 
\def\tup#1{\mathbf{#1}}

\newcommand*{\set}[1]{\{ #1 \}}
\newcommand*{\ranking}[1]{\langle #1 \rangle}
\newcommand*{\seq}[1]{\langle #1 \rangle}
\newcommand*{\lst}[2][m] {#2_1, \ldots, #2_{#1}}
\newcommand*{\defeq}{\ \mathrel{\mathop:}= \ }

\newcommand*{\iToj}{{i \rightarrow j}}
\newcommand*{\pluseq}{{\ \mathrel{{+}{=}}} \ }

\def\RIM{\mathsf{RIM}}
\def\mallows{\mathsf{MAL}}
\def\AMP{\mathsf{AMP}}
\def\LRIM{\RIM_{\mathsf{L}}}
\def\LMAL{\mallows_{\mathsf{L}}}

\definecolor{Gray}{gray}{0.5}
\definecolor{LightCyan}{rgb}{0.88,1,1}
\newcolumntype{a}{>{\columncolor{Gray}}c}
\def\w#1{{\cellcolor[gray]{0.5}\color{white}\textsf{#1}}}

\newcommand*{\ifff}{if and only if\xspace}
\newcommand*{\eg}{e.g.,\xspace}
\newcommand*{\ie}{i.e.,\xspace}

\def\Pr{\mathrm{Pr}}

\def\dist{\mathit{dist}}

\begin{document}
\title{Supporting Hard Queries over Probabilistic Preferences}

\numberofauthors{3}
\author{
\alignauthor Haoyue Ping\\
       \affaddr{New York University, USA}\\
       \email{hp1326@nyu.edu}
\alignauthor Julia Stoyanovich \titlenote{This work was supported in part by NSF Grants No. 1916647, 1926250, and 1934464.}
\\
       \affaddr{New York Univerisity, USA}\\
       \email{stoyanovich@nyu.edu}
\alignauthor Benny Kimelfeld \titlenote{This work was supported in part by BSF Grant No. 2017753 and ISF Grant No. 1295/15.}
\\
       \affaddr{Technion, Israel}\\
       \email{bennyk@cs.technion.ac.il}
}

\date{March 2020}

\maketitle

\begin{abstract}
Preference analysis is widely applied in various domains such as social choice and e-commerce. 
A recently proposed framework augments the relational database with a preference relation that represents uncertain preferences in the form of statistical ranking models, and provides methods to evaluate Conjunctive Queries (CQs) that express preferences among item attributes.
In this paper, we explore the evaluation of queries that are more general and harder to compute.

The main focus of this paper is on a class of CQs that cannot be evaluated by previous work.  These queries are provably hard since relate variables that represent items being compared.  To overcome this hardness, we instantiate these variables with their domain values, rewrite hard CQs as unions of such instantiated queries, and 
develop several exact and approximate solvers to evaluate these unions of queries.  
We demonstrate that exact solvers that target specific common kinds of queries are far more efficient than general solvers. 
Further, we demonstrate that sophisticated approximate solvers making use of importance sampling can be orders of magnitude more efficient than exact solvers, while showing good accuracy.  
In addition to supporting provably hard CQs, we also present methods to evaluate an important family of count queries, and of top-$k$ queries.
\end{abstract}

\section{Introduction}
\label{sec:intro}

Preferences are statements about the relative quality or desirability of items.  
Preference analysis aims to derive insight from a collection of preferences.
For example, in recommender systems~\cite{DBLP:journals/fcsc/BalakrishnanC12,DBLP:conf/wsdm/SarmaSGP10,DBLP:conf/webdb/StoyanovichJG15} and in political elections~\cite{Diaconis89,GormleyM06,gormley2008,mcelroy}, we may be interested in identifying the most preferred items or sets of items, or in understanding the points of consensus or disagreement among a group of voters.

Voter preferences are often inferred from indirect input (such as clicks on ads), or from  preferences of other similar voters based on demographic similarity or on similarity over stated preferences, as in collaborative filtering, and are thus uncertain.
A variety of statistical models have been developed to represent uncertain preferences~\cite{marden1995analyzing}, including the popular Mallows model~\cite{Mallows1957}.
There is much recent work in the machine learning and statistics communities~\cite{DBLP:conf/nips/AwasthiBSV14,BusseOB07,DBLP:journals/corr/DingIS15, GormleyM06,DBLP:journals/jair/HuangKG12,KamishimaA06,LebanonL02,DBLP:conf/nips/LebanonM07,DBLP:journals/jmlr/LuB14}, focusing specifically on learning the parameters of Mallows models or their mixtures~\cite{DBLP:journals/corr/DingIS15, Doignon2004, DBLP:conf/aaai/KenigIPKS18, DBLP:journals/jmlr/LuB14, DBLP:conf/webdb/StoyanovichIP16}.
Learning techniques for Mallows often rely on the Repeated Insertion Model (RIM)~\cite{Doignon2004} --- a generative model that gives rise to various distributions over  rankings.

In a recent work~\cite{DBLP:conf/pods/KenigKPS17}, we introduced a framework for representing and querying uncertain preferences in a \e{Probabilistic Preference Database}, or \e{PPD} for short.
We recall this framework here, illustrating it with an example. Consider Figure~\ref{fig:elections} that presents an instance of a polling database for the 2016 US presidential election. Each of
\textbf{Candidates} and \textbf{Voters} is an \e{ordinary relation} (abbr.~o-relation), while \textbf{Polls} is a \e{preference relation} (abbr.~p-relation) where each tuple is associated with a preference model---Mallows in this example.
Mallows models are ranking distributions parameterized by a center ranking $\bsigma$ and a dispersion parameter $\phi$.
We will discuss the Mallows model in Section~\ref{sec:preliminaries:rim}, explaining that it is a special case of RIM~\cite{Doignon2004}.
The PPD formalism of~\cite{DBLP:conf/pods/KenigKPS17}, on which we build here, accommodates RIM preferences, and we refer to such a database as a {\em RIM-PPD}.

In summary, a RIM-PPD represents uncertain preferences by statistical models.
Semantically, a RIM-PPD instance is a \e{probabilistic database}~\cite{DBLP:series/synthesis/2011Suciu}, where every random possible world (a deterministic database) is obtained by sampling from the stored RIM models.
RIM-PPDs adopt the conventional semantics of query evaluation over probabilistic databases, associating each answer with a \e{confidence value}---the probability of getting this answer in a random possible world~\cite{DBLP:series/synthesis/2011Suciu}.
Hence, query evaluation entails \e{probabilistic inference}: computing the marginal probability of query answers.
In the case of RIM-PPDs, query evaluation entails inference over statistical \e{ranking models}.

\begin{figure}[t!]
	\small
	\centering
	\begin{tabular}{|c c c c c c|}
		\multicolumn{4}{l}{{\bf Candidates} (o)}\\
		\hline
		\w{candidate} & \w{party} & \w{sex} & \w{age} & \w{edu} & \w{reg}\\
		\hline
		$\val{Trump}$ & $\val{R}$ & $\val{M}$ & $\val{70}$ & $\val{BS}$ & $\val{NE}$ \\
		$\val{Clinton}$ & $\val{D}$ & $\val{F}$ & $\val{69}$ & $\val{JD}$ & $\val{NE}$ \\
		$\val{Sanders}$ & $\val{D}$ & $\val{M}$ & $\val{75}$ & $\val{BS}$ & $\val{NE}$ \\
		$\val{Rubio}$ & $\val{R}$ & $\val{M}$ & $\val{45}$ & $\val{JD}$ & $\val{S}$ \\
		\hline
	\end{tabular}
	\quad
	\vskip0.3em
	\begin{tabular}{|c c c c|}
		\multicolumn{4}{l}{{\bf Voters} (o)}\\
		\hline
		\w{voter} & \w{sex} & \w{age} & \w{edu}\\
		\hline
		$\val{Ann}$ & $\val{F}$ & $\val{20}$ & $\val{BS}$\\
		$\val{Bob}$ & $\val{M}$ & $\val{30}$ & $\val{BS}$\\
		$\val{Dave}$ & $\val{M}$ & $\val{50}$ & $\val{MS}$\\
		\hline
	\end{tabular}
	\vskip0.3em
	\begin{tabular}{|c c|l}
		\multicolumn{3}{l}{{\bf Polls} (p)}\\
		\w{voter} & \w{date} & \underline{\textit{Preference model} $\mallows(\bsigma, \phi)$} \\
		\cline{1-2}
		$\val{Ann}$ & $\val{5/5}$ &  $\ranking{\val{Clinton},\val{Sanders},\val{Rubio},\val{Trump}}, 0.3$\\
		$\val{Bob}$ & $\val{5/5}$ & $\ranking{\val{Trump},\val{Rubio},\val{Sanders},\val{Clinton}}, 0.3$\\
		$\val{Dave}$ & $\val{6/5}$ & $\ranking{\val{Clinton},\val{Sanders},\val{Rubio},\val{Trump}}, 0.5$\\
		\cline{1-2}
	\end{tabular}
	\caption{An instance of RIM-PPD.}
    \label{fig:elections}
\end{figure}

A preference relation in a possible world represents a collection of orders, each called a {\em session}.
A tuple of a preference relation has the form $(\tup s;a;b)$, stating that in the order of session $\tup s$ item $a$ is preferred to item $b$, denoted $a \succ_{\tup s} b$.

For example, the tuple $(\val{Ann},\val{5/5};\val{Sanders};$ $\val{Clinton})$ in an instance of the \textbf{Polls} relation denotes that in a poll conducted on May $5^{th}$, Ann preferred Sanders to Clinton. Here, $(\val{Ann},\val{5/5})$ identifies a session. Note that the internal representation of a preference needs not store every pairwise comparison explicitly.  

Incorporating preferences into databases facilitates preference analysis. For example, an analyst may ask whether Ann prefers Trump to both Clinton and Rubio on May $5^{th}$ as follows, using $P$ to denote \textbf{Polls}:
\begin{align*} 
Q_0() \leftarrow & P(\val{Ann},\val{5/5};\val{Trump};\val{Clinton}), \\
                               & P(\val{Ann},\val{5/5};\val{Trump};\val{Rubio})
\end{align*}
 $Q_0$ is a Boolean conjunctive query (CQ) that computes the marginal probability of $\{\val{Trump} \succ \val{Clinton}, \val{Trump} \succ \val{Rubio}\}$ over the Mallows model of $(\val{Ann},\val{5/5})$.

The analyst may query preferences about the attributes of candidates, which generalizes the preferences over specific candidates.
For example, using $C$ to denote \textbf{Candidates}:
\[
\qOneIntro() \leftarrow P(\_,\_;c_1;c_2), C(c_1,\_,\val{F},\_,\_,\_), C(c_2,\_,\val{M},\_,\_,\_)
\]
The evaluation of $\qOneIntro$ computes the marginal probability that a female candidate is preferred to a male candidate over the random preferences of the users, drawn from their corresponding preference models.
We refer to the values of item attributes, such as \val{F} and \val{M}, as \e{labels}.
$\qOneIntro$ is an example of an \e{itemwise} CQ~\cite{DBLP:conf/pods/KenigKPS17}, querying preferences over labels.
Intuitively, itemwise CQs state a preference among constants and variables (\eg $c_1 \succ c_2$, or $c_1 \succ  \val{Trump}$) in addition to an independent condition on  item variables (\eg $c_1$ is a female candidate and $c_2$ is a male candidate), and this preference can be represented as a partial order of labels, named \e{label patterns} (\eg $\val{F} \succ \val{M}$).
Kenig et al.~\cite{DBLP:conf/pods/KenigKPS17} show that, at least for the fragment of queries without self-joins, itemwise CQs are {\em precisely} the queries that can be evaluated in polynomial time.
In a follow-up work, Cohen et al.~\cite{DBLP:conf/sigmod/CohenKPKS18} proposed a query engine that uses inference to evaluate these queries that have tractable complexity.

{\bf Problem statement.}
In this paper, we focus on extending RIM-PPD query evaluation to support general CQs, those that are {\em provably hard}.	
Given a non-itemwise CQ $Q$ and an instance $D$ of RIM-PPD, the goal is to calculate the probability that $Q$ holds in a random possible world.	This query evaluation problem is reduced to an inference problem over RIM.  We  investigate two types of queries beyond CQs, and also reduce their evaluation to inference over RIM.  This problem statement will be refined in Section~\ref{sec:overview:summary}.

To get the gist of our approach, consider the query:
\[
\qTwoIntro() \leftarrow P(\_,\_;c_1;c_2), C(c_1,\val{D},\_,\_,e,\_), C(c_2,\val{R},\_,\_,e,\_);
\]
$\qTwoIntro$ asks for the marginal probability that a Democrat $c_1$ is preferred to a Republican $c_2$ having the same education degree $e$. 
As $e$ is a variable, the qualified candidates for $c_1$ and $c_2$ cannot be determined ahead of time. 
According to the instance of \textbf{Candidates} in Figure~\ref{fig:elections}, $e$ takes on values $\val{BS}$ and $\val{JD}$.
Substituting $e$ with these values in $\qTwoIntro()$, we get:
\begin{align*} 
\qTwoIntro^{\val{BS}}() \leftarrow P(\_,\_;c_1;c_2), C(c_1,\val{D},\_,\_,\val{BS},\_), C(c_2,\val{R},\_,\_,\val{BS},\_);\\
\qTwoIntro^{\val{JD}}() \leftarrow P(\_,\_;c_1;c_2), C(c_1,\val{D},\_,\_,\val{JD},\_), C(c_2,\val{R},\_,\_,\val{JD},\_);
\end{align*}
Note that $\qTwoIntro^{\val{BS}}$ and $\qTwoIntro^{\val{JD}}$ are both itemwise CQs, and so their evaluation is tractable.  Further, according to the semantics of CQ evaluation, $\qTwoIntro$ holds if either $\qTwoIntro^{\val{BS}}$ holds or $\qTwoIntro^{\val{JD}}$ holds (\ie $\qTwoIntro = \qTwoIntro^{\val{BS}} \cup \qTwoIntro^{\val{JD}}$).	Note that it is possible for a ranking to satisfy both $\qTwoIntro^{\val{BS}}$ and $\qTwoIntro^{\val{JD}}$; $\ranking{\val{Sanders}, \val{Trump}, \val{Cliton}, \val{Rubio}}$ is an example.  Therefore, $\qTwoIntro^{\val{BS}}$ and $\qTwoIntro^{\val{JD}}$ are not mutually exclusive and $\Pr(\qTwoIntro) < \Pr(\qTwoIntro^{\val{BS}}) + 
\Pr(\qTwoIntro^{\val{JD}})$ may hold.

More generally, a non-itemwise CQ can be decomposed into a union of itemwise CQs, but the probability of a query union is not the sum of probabilities of its individual CQs.
The size of the union depends on the domain size of the instantiated variables.
We propose three exact solvers for the inference problem induced by this decomposition.
The first is  based on the inclusion-exclusion principle, and works for a union of any label patterns.
This solver, while general, does not scale well when the product of the domain sizes of the variables is large, and we use it as a performance baseline.  We propose two additional exact solvers, optimized for families of label patterns that are commonly used in practice: two-label patterns and bipartite patterns that are similar to bipartite graphs.

Further, we propose approximate solvers based on Multiple Importance Sampling (MIS).
We develop several flavors of approximate solvers, compare their performance, and show that they can outperform exact solvers by several orders of magnitude, while achieving good accuracy.

Finally,  we expand the family of supported queries to involve \e{Count-Session}, returning the number of sessions satisfying a given query $Q$, and \e{Most-Probable-Session}, returning $k$ sessions that support $Q$ with the highest probability.

\paragraph*{{\bf Contributions}} We make the following contributions: 
\begin{enumerate}
\item We reduce the evaluation of conjunctive queries over probabilistic preference databases to an inference problem over a union of label patterns (Section~\ref{sec:overview});
\item We develop exact solvers for CQs, Count-Session and Most-Probable-Session queries (Section~\ref{sec:exact}); 
\item We propose approximate solvers, based on Multiple Importance Sampling, that improve scalability, while  achieving good accuracy (Section~\ref{sec:approx}); and
\item We present results of an extensive experimental evaluation over real and synthetic datasets, demonstrating that (i) customized exact solvers see substantial improvement; (ii) approximate solvers are effective and scalable; (iii)  evaluation is well optimized for Most-Probable-Session queries; and (iv) the implementation can handle a large number of sessions  (Section~\ref{sec:experiments}).
\end{enumerate}
\section{Preliminaries}
\label{sec:preliminaries}

\subsection{Preferences and Label Patterns}
\label{sec:preliminaries:pref}

Let $A$ denote a set of $m$ items.
Preference is a binary relation over $A$. 
Let $a \succ b$ denote that $a \in A$ is preferred to $b \in A$. 
If the preference is from a judge $u$, we denote it by $a \succ_u b$. 
The preference relation $\succ$ is irreflexive, transitive, and asymmetric.

A \e{preference pair} compares two items. 
\e{Pairwise preferences} are a collection of preference pairs, such as $\{ a \succ b, a \succ c \}$.
They can be visualized by a directed graph with items as vertices and preference pairs as edges. 
If the directed graph is acyclic, it represents a \e{partial order}.
Since the relation $\succ$ is transitive, a partial order $\partialOrder$ expresses the same information as its transitive closure $tc(\partialOrder)$.

A \e{linear order} or \e{ranking} or \e{permutation} is a partial order where every two items in $A$ are comparable.
Let $\btau = \ranking{\lst{\tau}}$ denote a ranking placing item $\tau_i$ at rank $i$. 
We denote by $\btau(i)$ the item at rank $i$, by $\btau^{-1}(\tau)$ the rank of item $\tau$.
We denote by $rnk(A)$ the set of all $m!$ permutations over the items in $A$.
We denote by $\btau^k$ the truncated $\btau$ with only the first $k$ items, and by $A(\btau^k)$ the items in $\btau^k$.

A ranking $\btau$ is a \e{linear extension} of a partial order $\partialOrder$ if $\btau$ is consistent with $\partialOrder$ (\ie $\forall (x \succ y) \in \partialOrder, x \succ_{\btau} y$). 
We use $\Omega(\partialOrder)$ to denote the set of linear extensions of $\partialOrder$.

A \e{sub-ranking} $\subRanking$ is a ranking over a subset of the items in $A$, denoted by $A(\subRanking)$.
A sub-ranking can also be consistent with a partial order $\partialOrder$.
Let $\extension(\partialOrder)$ denote the set of sub-rankings that are consistent with $\partialOrder$, over the same set of items in $\partialOrder$, denoted by $A(\partialOrder)$.

\e{Labels} are values of item attributes.
For example, $\val{M}$ is a label of item \val{Trump} in Figure~\ref{fig:elections} that corresponds to the value of the \val{sex} attribute.
A \e{label pattern} (or just \e{pattern}) is a partial order of atomic labels or sets of labels.  For example,  $\set{\set{\val{M}, \val{JD}} \succ \val{BS}}$  denotes that male candidates with a JD are preferred to candidates with a BS degree.
A pattern can be represented by a directed acyclic graph $g$.
Figure~\ref{fig:pattern_example} presents a pattern $g_0 {=} \set{\val{F} {\succ} \val{M}}$ related to the RIM-PPD in Figure~\ref{fig:elections}.

\begin{figure}[t!]
	\centering
	\includegraphics[width=0.65\linewidth]{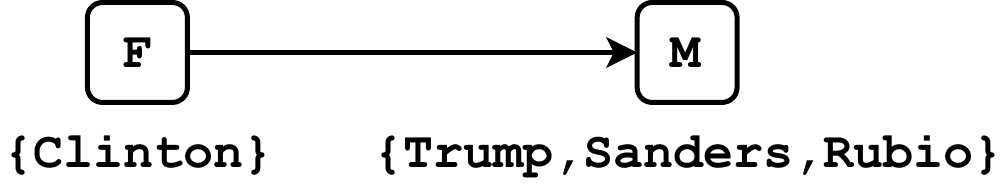}
	\caption{A label pattern over the polling database. Candidates with labels $\val{F}$ and $\val{M}$ are annotated below, respectively.}
	\label{fig:pattern_example}
\end{figure}

\subsection{Repeated Insertion Model}
\label{sec:preliminaries:rim}

The Repeated Insertion Model (RIM) is a generative ranking model that defines a probability distribution over permutations~\cite{Doignon2004}.
This distribution, denoted by $\RIM(\bsigma, \Pi)$, is parameterized by a reference ranking $\bsigma=\ranking{\lst{\sigma}}$ and a function $\Pi$, where $\Pi(i, j)$ is the probability of inserting $\sigma_i$ at position $j$.
Algorithm~\ref{alg:rim} presents the RIM sampling procedure.
It starts with an empty ranking, inserts items in the order of $\bsigma$, and puts item $\sigma_i$ at $j$-th position of the current incomplete ranking with probability $\Pi(i, j)$.

\begin{example}
	$\RIM(\ranking{a, b, c}, \Pi)$ generates $\btau' {=} \ranking{b, c, a}$ as follows.
	Initialize an empty ranking $\btau_0 {=} \ranking{}$. 
	At step 1, $\btau_1 {=} \ranking{a}$ by inserting $a$ into $\btau_0$ with probability $\Pi(1,1) {=} 1$. 
	At step 2, $\btau_2 {=} \ranking{b, a}$ by inserting $b$ into $\btau_1$ at position 1 with probability $\Pi(2,1)$. Note that $b$ is put before $a$ since $b \succ_{\btau'} a$.
	At step 3, $\btau' {=} \ranking{b, c, a}$ by inserting $c$ into $\btau_2$ at position 2 with probability $\Pi(3,2)$. 
	The overall probability of sampling $\btau'$ is $\Pr(\btau' \mid \ranking{a,b,c}, \Pi) {=} \Pi(1,1) \cdot \Pi(2,1) \cdot \Pi(3,2)$.
\end{example}

The Mallows model~\cite{Mallows1957}, $\mallows(\bsigma, \phi), \phi \in [0, 1]$, is a special case of RIM.
As a popular preference model, it defines a distribution of rankings that is analogous to the Gaussian distribution.
Ranking $\bsigma$ is at the center.
Rankings closer to $\bsigma$ have higher probabilities.
For a ranking $\btau$, its probability $\Pr(\btau|\bsigma, \phi) \propto \phi^{\dist(\bsigma, \btau)}$ where $\dist(\bsigma, \btau)$ is the Kendall-tau distance between $\bsigma$ and $\btau$: $\dist(\bsigma, \btau) = |{(a, a') | a \succ_{\bsigma} a', a' \succ_{\btau} a}|$ that is the number of disagreeing preference pairs.
When $\phi=0$, only $\bsigma$ has positive probability; 
when $\phi=1$, all rankings have the same probability, that is, $\mallows(\bsigma, 1)$ is the uniform distribution over rankings.
RIM was proposed in~\cite{Doignon2004} and provides an efficient and practical approach to draw rankings from the Mallows model.  This is because, as was shown in~\cite{Doignon2004}, $\RIM(\bsigma, \Pi)$ is precisely $\mallows(\bsigma, \phi)$ when $\Pi(i, j) = \frac{\phi^{i-j}}{1+\phi+...+\phi^{i-1}}$.

\begin{algorithm}[t!]
	\caption{RIM}
	\begin{algorithmic}[1]
		\REQUIRE $\RIM(\bsigma, \Pi)$, with $\bsigma = \ranking{\lst{\sigma}}$ \\
		\STATE Initialize an empty ranking $\btau = \ranking{}$.
		\FOR {$i = 1, \ldots, m$}
		\STATE Insert $\sigma_i$ into $\btau$ at $j \in [1, i]$ with probability $\Pi(i, j)$.
		\ENDFOR
		\RETURN  $\btau$
	\end{algorithmic}
	\label{alg:rim}
\end{algorithm}
	
The Approximate Mallows Posterior~\cite{DBLP:journals/jmlr/LuB14} $\AMP(\bsigma, \phi, \partialOrder)$, is a \e{sampler from the posterior distribution} of $\mallows(\bsigma, \phi)$ conditioned on a partial order $\partialOrder$.
When sampling a ranking, it follows the procedure of RIM, but the positions to insert items are constrained by $\partialOrder$.
Assume that $\btau_i$ is the current incomplete ranking when inserting $\sigma_i$.
Let $J$ denote the range of positions where inserting $\sigma_i$ does not violate $\partialOrder$.
Item $\sigma_i$ is inserted at $j \in J$ with probability $p_j\propto \phi^{i-j}$.

\begin{example}
	$\AMP(\ranking{a, b, c}, \phi, \set{c \succ a})$ generates ranking $\btau' {=}\ranking{b, c, a}$ as follows. 
	Initialize an empty ranking $\btau_0 {=} \ranking{}$. 
	At step 1, $\btau_1 {=} \ranking{a}$ by inserting $a$ into $\btau_0$. 
	At step 2, $\btau_2 {=} \ranking{b, a}$ by inserting $b$ at position 1 with probability $\frac{\phi}{1 + \phi}$.
	At step 3, $c$ must be placed before $a$, so $J {=} \set{1, 2}$. Consider that $p_1 \propto \phi^2$, $p_2 \propto \phi$, and $p_1 {+} p_2 {=} 1$. So $\btau' {=} \ranking{b, c, a}$ by inserting $c$ at position 2 with probability $p_2 {=} \frac{\phi}{\phi + \phi^2}$. The probability of sampling $\btau_0$ is $\Pr(\btau' \mid \ranking{a, b, c}, \phi, \set{c \succ a}) = \frac{\phi}{1 + \phi} \cdot \frac{\phi}{\phi + \phi^2} = \frac{\phi}{(1 + \phi)^2}$.
\end{example}

\subsection{Labeled RIM Matching}
\label{sec:preliminaries:label_matching}

We now  recall \e{labeled RIM matching}~\cite{DBLP:conf/pods/KenigKPS17}, an inference problem that  will be useful for query evaluation later.
A \e{labeled RIM}, denoted by $\LRIM(\bsigma,\Pi, \lambda)$, augments $\RIM(\bsigma,\Pi)$ with a labeling function $\lambda$, mapping each item to a finite set of its associated labels.
Let $\btau$ be a ranking of length $m$ generated by  $\LRIM(\bsigma,\Pi, \lambda)$.
An \e{embedding} of a label pattern $g$ in $\btau$ is a function $\embedding:nodes(g)\rightarrow [1, m]$  satisfying the conditions:
\begin{enumerate}
	\itemsep -0.3em
	\item Labels match: $\forall l\in nodes(g), l \in \lambda(\btau(\embedding(l)))$
	\item Edges match: $\forall (l, l') \in edges(g), \btau(\embedding(l)) \succ_{\btau} \btau(\embedding(l'))$
\end{enumerate}
If such embedding function $\embedding$ exists, we say that $\btau$ (w.r.t. $\lambda$) \e{matches} (or \e{satisfies}) $g$, denoted by $(\btau, \lambda) \models g$.
When $\lambda$ is clear from context, we write  $\btau \models g$.
The items selected by the embedding function are the \e{matching items}.

\begin{example}
	Given a ranking $\btau_0 = \langle \val{Trump}, \val{Clinton},$ $\val{Sanders}, \val{Rubio} \rangle$, the labeling function $\lambda_0$ in Figure~\ref{fig:elections}, and the pattern $g_0$ in Figure~\ref{fig:pattern_example}, there exists an embedding function $\embedding_0 = \set{\val{F} \mapsto 2, \val{M} \mapsto 3}$, with matching items $\btau_0(2){=}\val{Clinton}$ for label \val{F}, and $\btau_0(3){=}\val{Sanders}$ for \val{M}.
	The edge $(\val{F}, \val{M})$ matches $\val{Clinton} \succ_{\btau_0} \val{Sanders}$,
	and so $(\btau_0, \lambda_0) \models g_0$ with $\embedding_0$.
\end{example}

The problem of \e{pattern matching} on labeled RIM is as follows.
Given $\LRIM(\bsigma, \Pi, \lambda)$ and a pattern $g$, compute the probability that a random ranking $\btau \sim \RIM(\bsigma, \Pi)$ satisfies $g$ (w.r.t. $\lambda$).
This is also the marginal probability of $g$ over $\LRIM(\bsigma, \Pi, \lambda)$:
\begin{equation}
	\Pr(g \mid \bsigma, \Pi, \lambda) = \sum_{\substack{\btau \in rnk(A) \\ (\btau, \lambda) \models g}} {\Pr(\btau \mid \bsigma, \Pi)}
\end{equation}
where $rnk(A)$ is the set of all $m!$ rankings over items $A$.

\section{Query Evaluation}
\label{sec:overview}

In this section, we explain query evaluation in a RIM-PPD and refine the problem statement given in Section~\ref{sec:intro}.

\subsection{Conjunctive Query Evaluation}
\label{sec:overview:cq}

Given a Conjunctive Query (CQ) expressing preferences with a p-relation, if all atoms of p-relation refer to the same session, this query is a \e{sessionwise} CQ.
If the sessionwise CQ is equivalent to a label pattern over each session, this is an \e{itemwise} CQ. 
Otherwise, a \e{non}-\e{itemwise} CQ.

In a recent paper, we showed how to reduce query evaluation of itemwise CQs to labeled RIM matching, and developed a solver for this inference problem, called Lifted Top Matching (LTM)~\cite{DBLP:conf/sigmod/CohenKPKS18}.
Given an itemwise CQ $Q$ and a RIM-PPD $D$, we wish to compute the marginal probability that $Q$ is satisfied.
Under the assumption that there are $n$ independent sessions $\set{\lst[n]{\tup{s}}}$ in a p-relation, we can evaluate $Q$ over each session and aggregate the results from all sessions as follows:
\[
\Pr(Q \mid D) = 1 - \prod_{i=1}^{n}(1 - \Pr(Q \mid \tup{s}_i))
\]

Thus, query evaluation is reduced to evaluating the query over each session.
For a particular session $\tup{s}$, we denote by $\RIM(\bsigma^\tup{s}, \Pi^\tup{s})$ its RIM model, by $\lambda$ the labeling function of database $D$, and by $g$ the label pattern corresponding to $Q$ (as defined in Section~\ref{sec:preliminaries:pref}), which leads to the labeled RIM matching problem in Section~\ref{sec:preliminaries:label_matching}.
Let $\LRIM(\bsigma^\tup{s}, \Pi^\tup{s}, \lambda)$ denote the labeled RIM over session $\tup{s}$.
The probability that $Q$ holds on session $\tup{s}$ is the marginal probability of $g$ over $\LRIM(\bsigma^\tup{s}, \Pi^\tup{s}, \lambda)$.
\[
\Pr(Q \mid \tup{s}) = \Pr(g \mid \bsigma^\tup{s}, \Pi^\tup{s}, \lambda) = \sum_{\substack{\btau \in rnk(A) \\ (\btau, \lambda) \models g}}  {\Pr(\btau \mid \bsigma^\tup{s}, \Pi^\tup{s})}
\]
LTM calculates this probability with complexity $O(2^q m^q)$, where $q$ is the number of nodes in $g$, see~\cite{DBLP:conf/sigmod/CohenKPKS18} for details.

Non-itemwise CQs are the sessionwise CQs with some variable(s) preventing label pattern reduction.
In contrast to itemwise CQs, for which query evaluation has polynomial-time data complexity, the evaluation of non-itemwise CQs is \#P-hard~\cite[Theorems 4.4 and 4.5]{DBLP:conf/sigmod/CohenKPKS18}.
To evaluate a non-itemwise CQ, we ground its variables, and rewrite it into a union of itemwise CQs.
Let $\ungroundedVariables(Q)$ denote the set of variables to ground.
Algorithm~\ref{alg:decomposition} decomposes a non-itemwise CQ $Q$ into a union of itemwise CQs by grounding these variables in $\ungroundedVariables(Q)$.
For example, $\qTwoIntro$ in Section~\ref{sec:intro} is non-itemwise due to variable $e$.
So $\ungroundedVariables(\qTwoIntro)=\{e\}$ and $\qTwoIntro = \qTwoIntro^{\val{BS}} \cup \qTwoIntro^{\val{JD}}$.
Note that these CQs are neither disjoint nor independent.
For each session in a RIM-PPD, a union of itemwise CQs is equivalent to a union of label patterns, and the probability of $Q$ is the sum of the probabilities of rankings that satisfy at least one pattern in the union.

\begin{algorithm}[t!]
	\caption{DecomposeQuery}
	\begin{algorithmic}[1]
	\REQUIRE Database $D$, non-itemwise query $Q$\\
		\STATE Calculate $\ungroundedVariables(Q)$, the set of variables to ground.
        \STATE Calculate $Doms$, the domains of $\ungroundedVariables(Q)$ in $D$.
        \STATE $U = \emptyset$
        \FOR {$\nu$ in CartesianProduct($Doms$)}
            \STATE $\nu$ maps each variable to a value in its domain.
            \STATE Generate $Q_{\nu}$ by instantiating $Q$ with $\nu$.
            \STATE $U = U \cup Q_{\nu}$
		\ENDFOR
		\RETURN $U$
	\end{algorithmic}
	\label{alg:decomposition}
\end{algorithm}

\subsection{Beyond Conjunctive Queries}
\label{sec:overview:count}

\medskip
\noindent\textbf{Count-Session.} A Boolean CQ $Q$ computes the probability that $Q$ is satisfied in a random possible world, while a Count-Session query, denoted $\countq(Q)$, computes the number of sessions satisfying $Q$.
Since RIM-PPDs are probabilistic, $\countq(Q)$ is evaluated under the possible world semantics, and corresponds to the {\em expectation} of $\countq(Q)$ over the distribution of possible worlds. 

Let $S = \{ \lst[n]{\tup{s}}\}$ denote $n$ sessions in a p-relation. The expectation of $\countq(Q)$ is the sum of the probabilities that the sessions satisfy $Q$: $\countq(Q) = \sum_{i=1}^{n} \Pr(Q|\tup{s_i})$.

\medskip
\noindent\textbf{Most-Probable-Session.}
For a Boolean CQ $Q$ and an integer $k$, a Most-Probable-Session query, denoted $\topk(Q, k)$, finds $k$ sessions in which $Q$ is satisfied with the highest probability.
We implement two strategies for this operator.
The first calculates $\Pr(Q)$ for each session, then selects $k$ most supportive sessions.
The second strategy, named \e{top-$k$ optimization}, first quickly calculates the upper bounds for all sessions, and then calculates the exact probability of sessions in descending order of their upper bounds,  stopping once there are at least $k$ sessions whose exact probability is no lower than the highest remaining upper-bound. 

We  will present an approach to compute the upper-bound of any pattern union using a {\em bipartite solver} that  implements the top-$k$ optimization in  Section~\ref{sec:approx:ub}.
This approach constructs a new pattern union $G'$ with selected edges from the original $G$. To derive a tight upper-bound, we want to keep the edges that are hardest to satisfy.	We first calculate all possible edges in $G$ by transitive closure, then select edges using the following heuristic.

Let $\alpha(l \mid \btau)$ be the minimum position (highest rank) of items with label $l$ in a ranking $\btau$, and let $\beta(l \mid \btau)$ be the maximum position (lowest rank).
The {\em ease} of an edge $(l, l')$ to be satisfied by a random permutation from $\mallows(\bsigma, \phi)$ is estimated by:
\[\mathtt{ease}(l, l' \mid \bsigma) = \beta( l' \mid \bsigma) - \alpha( l \mid \bsigma)\]
We construct $G'$ with edges of small $\mathtt{ease}$ values. 
If only one edge is selected for each pattern, $G'$ is a union of two-label patterns, and $\topk(Q, k)$ invokes the two-label solver (see Section~\ref{sec:exact:2label}).
Otherwise, $G'$ is a union of bipartite patterns, and the bipartite solver is invoked (see Section~\ref{sec:exact:bipartite}).

Exact solvers have complexity exponential in the number of labels, so $G'$ is much faster to compute.  Because fewer labels and fewer edges lead to fewer constraints,  more permutations  satisfy $G'$, and so $\Pr(G'|\bsigma, \Pi) \geq \Pr(G|\bsigma, \Pi)$.

\subsection{Problem Statement}
\label{sec:overview:summary}
Queries in this paper include non-itemwise CQs, Count-Session queries, and Most-Probable-Session queries.
The evaluation of these hard queries is reduced to a generalized inference problem of labeled RIM matching: given a pattern union $G = g_1 \cup \ldots \cup g_\numPatterns$, compute its marginal probability over $\LRIM(\bsigma, \Pi, \lambda)$:
\begin{equation}
\label{eq:inference_of_pattern_union}
\Pr(G \mid \bsigma, \Pi, \lambda) = \sum_{\substack{\btau \in rnk(A) \\ \exists g \in G, (\btau, \lambda) \models g}}  {\Pr(\btau \mid \bsigma, \Pi)}
\end{equation}
Sections~\ref{sec:exact} and~\ref{sec:approx} will present exact and approximate solvers for this problem, respectively.
\section{Exact solvers}
\label{sec:exact}

Let $\LRIM(\bsigma, \Pi, \lambda)$ be a labeled RIM model with reference ranking $\bsigma=\ranking{\lst{\sigma}}$.
Let $G = g_1 \cup \ldots \cup g_\numPatterns$ be a union of $\numPatterns$ patterns.
We are interested in the marginal probability of $G$ over $\LRIM(\bsigma, \Pi, \lambda)$ defined in Equation~\eqref{eq:inference_of_pattern_union}.

Equation~\eqref{eq:inference_of_pattern_union} needs to enumerate $m!$ permutations. In this section, we will propose more efficient approaches.

\subsection{General Solver}
\label{sec:exact:general}

The general solver applies inclusion-exclusion principle:
\begin{equation}
\label{eq:inclusion_exclusion_principle}
\begin{split}
\Pr(G \mid \bsigma, \Pi, \lambda) 
& = \Pr(g_1 \cup \ldots \cup g_\numPatterns \mid \bsigma, \Pi, \lambda) \\
& =  \sum_{i=1}^\numPatterns \Pr(g_i \mid \bsigma, \Pi, \lambda) \\
& - \sum_{1 \leq i_1 < i_2 \leq \numPatterns} \Pr(g_{i_1} \wedge g_{i_2} \mid \bsigma, \Pi, \lambda) \\
& + \ldots  \\
& + (-1)^{(\numPatterns-1)}\Pr(g_1 \wedge \ldots \wedge g_\numPatterns \mid \bsigma, \Pi, \lambda)
\end{split}
\end{equation}
where the conjunction $g_i \wedge \ldots \wedge g_j$ is a pattern containing all nodes and edges in $\set{g_i, ..., g_j}$.

\begin{example}
	Let $G = g_1 \cup g_2$ where $g_1 = \set{l_1 \succ l_2}$ and $g_2 = \set{l_3 \succ l_4}$. Its marginal probability over $\LRIM(\bsigma, \Pi, \lambda)$ is $\Pr(g_1 \mid \bsigma, \Pi, \lambda) + \Pr(g_2 \mid \bsigma, \Pi, \lambda) - \Pr(g_3 \mid \bsigma, \Pi, \lambda)$ where $g_3 = g_1 \wedge g_2 = \set{l_1 \succ l_2, l_3 \succ l_4}$.
\end{example}

The RIM inference problem for pattern unions has been reduced to a RIM inference problem for patterns, which can be solved by the LTM solver~\cite{DBLP:conf/sigmod/CohenKPKS18}. 
The complexity of LTM is $O(2^q m^q)$, where $m$ is the number of items in $\bsigma$ and $q$ is the number of nodes in one pattern~\cite{DBLP:conf/sigmod/CohenKPKS18}. 
The complexity of the general solver is dominated by the largest  pattern conjunction $g_1 \wedge g_2 \wedge \ldots \wedge g_\numPatterns$. 
Assuming that each $g_i$ has $q$ nodes, the general solver runs in $O((2m)^{q \cdot \numPatterns})$.  
We use this solver as a baseline in our experiments.

\subsection{Two-label Solver}
\label{sec:exact:2label}

A  common class of queries concerns analysis of preferences over a pair of items.  Such queries are reduced to a union of \e{two-label} patterns, and we call them \e{two-label} queries.
For example, $\qTwoIntro$ in Section~\ref{sec:intro} is a two-label query: $\qTwoIntro() \leftarrow P(\_,\_;c_1;c_2), C(c_1,\val{D},\_,\_,e,\_), C(c_2,\val{R},\_,\_,e,\_)$.
By instantiating $e$ with $\val{BS}$ and $\val{JD}$, $\qTwoIntro$ is reduced to a pattern union $G = g_1 \cup g_2$, where $g_1 = \{\{\val{D}, \val{BS}\} \succ \{\val{R}, \val{BS}\}\}$ and $g_2 = \{\{\val{D}, \val{JD}\} \succ \{\val{R}, \val{JD}\}\}$ are both two-label patterns. 

Since all patterns in $G$ only have two labels, we re-write $G = g_1 \cup \ldots \cup g_\numPatterns = \bigcup_{i=1}^\numPatterns \{ l_i \succ r_i \}$.
The labels $\set{\lst[\numPatterns]{l}}$ are the L-type labels, while $\set{\lst[\numPatterns]{r}}$ R-type.

Instead of calculating the probability that $G$ is satisfied, the two-label solver calculates the probability that $G$ is violated. Let $\btau \not\models g$ and $\btau \not\models G$ denote that a permutation $\btau$ violates a pattern $g$ and a pattern union $G$, respectively. Then $\btau \not\models G$ \ifff $\forall g_i \in G, \btau \not\models g_i$.
Let $\alpha(l)$ be the minimum position (highest rank) of items with label $l$ in a ranking, while $\beta(l)$ the maximum position (lowest rank).
These are the Min/Max positions of a label in a ranking.
Given a two-label pattern $g = \{ l \succ r \}$ and a ranking $\btau$, we can check whether $\btau \models g$ by the Min/Max positions of labels.  Namely, $\btau \models g$ if $\alpha(l) < \beta(r)$ and $\btau \not\models g$ if $\alpha(l) \geq \beta(r)$.

Algorithm~\ref{alg:two_label} presents the two-label solver.
It first calculates the complementary event of $G$ by dynamic programming during RIM insertions.
States are in the form of $\seq{\alpha, \beta}$, tracking Min positions for L-type labels and Max positions for R-type labels.
States in $\mathcal{P}_i$ are generated by inserting item $\sigma_i$ into the states in $\mathcal{P}_{i-1}$.
Let $\seq{\alpha_\iToj, \beta_\iToj}$ denote a new state generated by inserting item $\sigma_i$ into $\seq{\alpha, \beta}$ at position $j$; 
$\seq{\alpha_\iToj, \beta_\iToj}$ is updated from $\seq{\alpha, \beta}$ as follows:
\begin{itemize}
    \itemsep -0.3em
    \item $\alpha_\iToj(l) = min(\alpha(l), j)$ if $l \in \lambda(\sigma_i)$ and $l$ is L-type;
    \item $\beta_\iToj(l) = max(\beta(l), j)$ if $l \in \lambda(\sigma_i)$ and $l$ is R-type;
    \item $\alpha_\iToj(l) = \alpha(l) + 1$ if $l \notin \lambda(\sigma_i)$ and $\alpha(l) \geq j$;
    \item $\beta_\iToj(l) = \beta(l) + 1$ if $l \notin \lambda(\sigma_i)$ and $\beta(l) \geq j$.
\end{itemize}
The algorithm only tracks the states that violate $G$, and its complexity is $O(m^{2\numPatterns+1})$.

\begin{algorithm}[t!]
	\caption{TwoLabelSolver}
	\begin{algorithmic}[1]
		\REQUIRE $\LRIM(\bsigma, \Pi, \lambda)$, $G = \bigcup_{i=1}^\numPatterns \{ l_i \succ r_i \}$ \\
		\STATE $\mathcal{P}_0 {\defeq} \set{\seq{\set{}, \set{}}}, q_0 {\defeq} \set{\seq{\set{}, \set{}} \mapsto 1}$
		\FOR {$i=1, .., m$}
			\STATE $\mathcal{P}_i \defeq \set{}$
			\FOR {$\seq{\alpha, \beta} \in \mathcal{P}_{i-1}$}
				\FOR {$j = 1, ..., i$}
					\STATE Generate a new state $\seq{\alpha_\iToj, \beta_\iToj}$ by inserting $\bsigma_i$ into $\seq{\alpha, \beta}$ at position $j$, and updating Min/Max positions according to the labeling function $\lambda$.
					\IF {$\seq{\alpha_\iToj, \beta_\iToj} \not\models G$}
						\STATE $\mathcal{P}_i.add(\seq{\alpha_\iToj, \beta_\iToj})$
						\STATE $q_i(\seq{\alpha_\iToj, \beta_\iToj}) \pluseq  q_{i-1}(\seq{\alpha, \beta}) \cdot \Pi(i, j)$ \label{alg:two_label:q_i}
					\ENDIF
				\ENDFOR
			\ENDFOR
		\ENDFOR
		\RETURN  $1 - \sum_{\seq{\alpha, \beta} \in \mathcal{P}_m} {q_m(\seq{\alpha, \beta})}$
	\end{algorithmic}
	\label{alg:two_label}
\end{algorithm}

\begin{example}
	Let $\LRIM(\bsigma_0, \Pi_0, \lambda_0)$ be a labeled RIM with $\bsigma_0 = \ranking{a,b,c}$.
	Let $G = g_1 \cup g_2$ be a pattern union. We will focus on $g_1$ in this example.
	Let $g_1 = \set{l_1 \succ r_1}$.
	Assume that $\lambda_0$ associates items $a$ and $c$ with label $l_1$, and $b$ with label $r_1$.
	At step 1, insert $a$ and generate state $\seq{\alpha_1, \beta_1}$ with probability $q_1(\seq{\alpha_1, \beta_1}) = \Pi_0(1,1) = 1$, where $\alpha_1 = \set{l_1 \mapsto 1}$ and $\beta_1 = \set{}$.
	At step 2,  $b$ must be inserted before $a$ to violate $g_1$.
	So $\alpha_2 = \set{l_1 \mapsto 2}$, $\beta_2 = \set{r_1 \mapsto 1}$, and $q_2(\seq{\alpha_2, \beta_2}) = q_1(\seq{\alpha_1, \beta_1}) \cdot \Pi_0(2,1) = \Pi_0(2,1)$.
	At step 3,  $c$ must be inserted after $b$ to violate $g_1$.
	So $\beta_3 = \set{r_1 \mapsto 1}$.
	Item $c$ can be inserted either before item $a$ generating $\alpha_3(l_1) = min(\alpha_2(l_1), 2)=2$ with probability $\Pi_0(3,2)$, or after $a$ generating $\alpha_3(l_1) = min(\alpha_2(l_1), 3)=2$ with probability $\Pi_0(3,3)$. 
	Both scenarios generate the same $\alpha_3(l_1)=2$, thus their probabilities are merged by $q_3(\seq{\alpha_3, \beta_3}) = q_2(\seq{\alpha_2, \beta_2}) \cdot \Pi_0(3, 2) + q_2(\seq{\alpha_2, \beta_2}) \cdot \Pi_0(3, 3) = \Pi_0(2,1) \cdot (\Pi_0(3, 2) + \Pi_0(3, 3))$.
\end{example}

\begin{theorem}
Given $\LRIM(\bsigma, \Pi, \lambda)$ and a union of two-label patterns $G$, Algorithm~\ref{alg:two_label} returns $\Pr(G \mid \bsigma, \Pi, \lambda)$, the marginal probability of $G$ over $\LRIM(\bsigma, \Pi, \lambda)$.
\end{theorem}

\begin{proof}
	According to Equation~\ref{eq:inference_of_pattern_union}, the marginal probability of $G$ over $\LRIM(\bsigma, \Pi, \lambda)$ is the sum of the probabilities of all rankings that satisfy $G$.
	Algorithm~\ref{alg:two_label} first calculates its negation that is the sum of the probabilities of all rankings that violate $G$.
		
	Recall that $G = g_1 \cup \ldots \cup g_\numPatterns = \bigcup_{i=1}^\numPatterns \{ l_i \succ r_i \}$, and a ranking $\btau \not\models G$ if and only if $\forall g_i \in G, \btau \not\models g_i$.
	From the perspective of Min/Max conditions, $\btau \not\models g_i$ means that $\alpha(l_i) \geq \beta(r_i)$.
	As a result, Algorithm~\ref{alg:two_label} only tracks the Min/Max positions of labels for the generated rankings during RIM insertions, and groups rankings sharing the same Min/Max positions of labels into a state $\seq{\alpha, \beta}$.
	At step $i$ of RIM insertions, $\mathcal{P}_i$ is the set of states that violate $G$, and $q_i(\seq{\alpha, \beta})$ represents the sum of probabilities of the generated rankings of length $i$ included in $\seq{\alpha, \beta}$.
	Note that once a state can satisfy $G$ at step $i$, it will always satisfy $G$ in the future with the same matching items at step $i$.
	So the algorithm only tracks states that violate $G$, and prunes states that satisfy $G$.
	We prove correctness of Algorithm~\ref{alg:two_label} by induction.
	
	The algorithm starts with an empty state $\seq{\set{}, \set{}}$, since no item is inserted yet.
	It is associated with the probability 1, meaning that no ranking or state was pruned yet.
	
	At step 1, item $\sigma_1$ is inserted into an empty ranking represented by $\seq{\set{}, \set{}}$ at position 1 with probability $\Pi(1,1)=1$.
	A state $\seq{\alpha_1, \beta_1}$ is generated, and $q_1(\seq{\alpha_1, \beta_1}) = 1$.
	If $\lambda(\sigma_1) = \emptyset$, $\seq{\alpha_1, \beta_1}  = \seq{\set{}, \set{}}$; Otherwise, $\forall l \in \lambda(\sigma_1)$, $\alpha_1(l) = 1$ if $l$ is L-type, and $\beta_1(l) = 1$ if $l$ is R-type.
	Only one ranking $\ranking{\sigma_1}$ is generated at step 1 and it cannot satisfy $G$.
	So the state in $\mathcal{P}_1 = \set{\seq{\alpha_1, \beta_1}}$ includes every ranking over the first item $\sigma_1$ that violates $G$.
	
	At step $i$, the algorithm reads states from $\mathcal{P}_{i-1}$ and their probabilities from $q_{i-1}$.
	These states are over the first $(i-1)$ items in $\bsigma$, denoted by $A(\bsigma^{i-1})$. 
	Assume that the states in $\mathcal{P}_{i-1}$ include all rankings over $A(\bsigma^{i-1})$ that violate $G$, and that the probabilities in $q_{i-1}$ are correct.
	Note that any ranking $\btau$ over $A(\bsigma^i)$ violating $G$ can be generated by inserting $\sigma_i$ into $\btau_{-\sigma_i}$, a ranking with $\sigma_i$ removed from $\btau$, over $A(\bsigma^{i-1})$ that also violates $G$.
	Inserting $\sigma_i$ into every state of $\mathcal{P}_{i-1}$ at every possible position $j$ will generate all states required by $\mathcal{P}_i$.
	
	Let $\seq{\alpha_\iToj, \beta_\iToj}$ denote the new state generated by inserting $\sigma_i$ at position $j$ into state $\seq{\alpha, \beta} \in \mathcal{P}_{i-1}$. 
	The values of $\alpha_\iToj$ and $\beta_\iToj$ should be updated according to the algorithm description in order to reflect the Min/Max positions correctly, so that the algorithm can determine whether $\seq{\alpha_\iToj, \beta_\iToj}$ satisfies $G$.
	If so, this state is pruned.
	Otherwise, it is added into $\mathcal{P}_i$, and its probability is also tracked by $q_i$.
	Recall that $\seq{\alpha, \beta}$ represents a collection of rankings of the same Min position mappings $\alpha$ and Max position mappings $\beta$, and $q_{i-1}(\seq{\alpha, \beta})$ is the sum of the probabilities of these rankings.
	
	Assume that there are $N$ rankings $\set{\lst[N]{\btau}}$ in this collection. Then 
	\[q_{i-1}(\seq{\alpha, \beta}) = \sum_{k=1}^{N} {\Pr(\btau_k)}\]
	and
	\begin{equation*}
	\begin{split}
	\Pr(\seq{\alpha_\iToj, \beta_\iToj}) 
	& = \sum_{k=1}^{N} {\big( \Pr(\btau_k) \cdot \Pi(i, j) \big)}\\
	& = q_{i-1}(\seq{\alpha, \beta}) \cdot \Pi(i, j)
	\end{split}
	\end{equation*}
	Note that multiple states in $\mathcal{P}_{i-1}$ may generate the same new state when inserting $\sigma_i$ at different positions.
	So Algorithm~\ref{alg:two_label} accumulates $\Pr(\seq{\alpha_\iToj, \beta_\iToj})$ into $q_i(\seq{\alpha_\iToj, \beta_\iToj})$ (Line~\ref{alg:two_label:q_i}).
	
	After iterating all states in $\mathcal{P}_{i-1}$ and all positions $j \in \set{1, \ldots, i}$, $\mathcal{P}_i$ includes all states that are over $A(\bsigma^i)$ and violate $G$, and the probabilities in $q_i$ are also correct.
	
	At step $m$, all items are inserted, so all rankings that violate $G$ have been included in the states of $\mathcal{P}_m$. Then $\Pr(G|\bsigma, \Pi) = 1 - \sum_{\seq{\alpha, \beta} \in \mathcal{P}_m} {q_m(\seq{\alpha, \beta})}$.
\end{proof}

\subsection{Bipartite Solver}
\label{sec:exact:bipartite}

A bipartite pattern is similar to a bipartite graph.
The nodes are classified into two sets $L$ and $R$, such that all directed edges are in the form $(l, r), l \in L, r \in R$.
Labels in $L$ and $R$ are L-type and R-type, respectively.

With the definition of $\alpha$ and $\beta$ in Section~\ref{sec:exact:2label}, an edge $(l, r)$ in a bipartite pattern is essentially $\alpha(l) < \beta(r)$.
A ranking satisfies a bipartite pattern $g$ if it satisfies all Min/Max constraints specified by $g$.

For a union of bipartite patterns $G = g_1 \cup \ldots \cup g_\numPatterns$, the  solver tracks $\alpha$ for L-type labels and $\beta$ for R-type labels. A permutation satisfies $G$ if it satisfies any pattern $g \in G$.

\subsubsection{Algorithm Description}
The basic version of a bipartite solver works as follows.
It is a Dynamic Programming  algorithm that tracks the minimum positions of L-type labels and the maximum positions of R-type labels, during RIM insertion process.
At step $i$, the first $i$ items in $\bsigma$ are inserted, and $i!$ rankings are generated accordingly.
These rankings are grouped into states in the form of $\seq{\alpha, \beta}$ where $\alpha$ maps L-type labels to their minimum positions and $\beta$ maps R-type labels to their maximum positions.
After all items are inserted, enumerate all states and add up the probabilities of the states satisfying at least one pattern $g_i \in G$.
The complexity of this algorithm is $O(m^{q \numPatterns})$, where $m$ is the number of items in $\bsigma$, $q$ is the number of labels per pattern, and $\numPatterns$ is the number of patterns in $G$.

The more sophisticated version of bipartite solver dynamically prunes labels tracked by states based on the ``situations'' of patterns and edges. The ``situations'' are \{satisfied, violated, uncertain\}.
An edge $(l, r)$ is satisfied if $\alpha(l) < \beta(r)$; violated if $\alpha(l) \geq \beta(r)$ after all items in $l$ and $r$ are inserted; uncertain if it is neither satisfied nor violated.
A pattern is satisfied if all its edges are satisfied; violated if any of its edges are violated; and uncertain otherwise.

The key observation is that once an edge is satisfied by a state, this state will always satisfy this edge in the future. The same is true for an edge being violated, a pattern being satisfied, and a pattern being violated. This  enables several optimization opportunities:
\begin{itemize}
\itemsep -0.3em
\item An edge is satisfied: no need to track this edge.
\item An edge is violated: the entire pattern is violated, no need to track this pattern.
\item A pattern is satisfied: the pattern union $G$ is satisfied, add the probability of this state into the marginal probability, no need to track this pattern.
\item A pattern is violated: no need to track this pattern.  
\end{itemize}

\textbf{In summary}, the bipartite solver only needs to track labels in uncertain edges of uncertain patterns.

\begin{algorithm}[t!]
	\caption{BipartiteSolver}
	\begin{algorithmic}[1]
		\REQUIRE $\LRIM(\bsigma, \Pi, \lambda)$, $G = g_1 \cup \ldots \cup g_\numPatterns$ \\
		\STATE $\mathcal{P}_0 {\defeq} \set{\seq{\set{}, \set{}}}$, $\mathcal{E}_0 {\defeq} \set{\seq{\set{}, \set{}} \mapsto G}$, \\ $q_0 {\defeq} \{ \seq{\set{}, \set{}} \mapsto 1 \}$
        \STATE $prob \defeq 0$
		\FOR {$i=1, .., m$}
    		\STATE $\mathcal{P}_i \defeq \set{}$
    		\FOR {$\seq{\alpha, \beta} \in \mathcal{P}_{i-1}$}
                \STATE $G_u \defeq \mathcal{E}_{i-1}(\seq{\alpha, \beta})$
        		\FOR {$j = 1, ..., i$}
        		\STATE Generate a new state $\seq{\alpha_\iToj, \beta_\iToj}$ by inserting $\bsigma_i$ into $\seq{\alpha, \beta}$ at position $j$, and updating Min/Max positions according to the labeling function $\lambda$.
        		\IF {$\seq{\alpha_\iToj, \beta_\iToj}$ violates all patterns in $G_u$}
        		  \STATE Ignore $\seq{\alpha_\iToj, \beta_\iToj}$
        		\ELSE
            		\STATE $p' \defeq q_{i-1}(\seq{\alpha, \beta}) \cdot \Pi(i, j)$
            		\STATE $G_u' \defeq$ OnlyTrackUncertainPatterns($G_u$)
            		\IF {$\exists g \in G_u'$, all edges in $g$ are satisfied}
                        \STATE $prob \pluseq p'$
            		\ELSE
                		\STATE $\seq{\alpha_\iToj, \beta_\iToj}$.onlyTrackLabelsFor($G_u'$)
                		\STATE $\mathcal{P}_i.add(\seq{\alpha_\iToj, \beta_\iToj})$
                        \STATE $\mathcal{E}_i(\seq{\alpha_\iToj, \beta_\iToj}) \defeq G_u'$
                        \STATE $q_i(\seq{\alpha_\iToj, \beta_\iToj}) \pluseq p'$
            		\ENDIF
        		\ENDIF
        		\ENDFOR
    		\ENDFOR
		\ENDFOR
		\RETURN  $prob$
	\end{algorithmic}
	\label{alg:bipartite}
\end{algorithm}

Algorithm~\ref{alg:bipartite} presents the bipartite solver that uses RIM (see Section~\ref{sec:preliminaries:rim}) as basis for inference.
At step $i$, it maintains a set of states $\mathcal{P}_i$.
A state $\seq{\alpha, \beta}$ tracks the Min/Max positions of labels.
The $\mathcal{E}_i$ maps a state $\seq{\alpha, \beta}$ to $G_u$, a union of uncertain patterns with uncertain edges in this state. Before running RIM, all patterns and edges are uncertain, so $G_u {=} G$.
The probabilities of the states are tracked by $q_i$.  

Recall that RIM sampling starts with an empty ranking.
Therefore, the initial state is $\seq{\set{}, \set{}}$, and $\mathcal{E}_0(\seq{\set{}, \set{}}) = G$, $q_0(\seq{\set{}, \set{}})=1$.
At step $i$, generate new states by inserting item $\sigma_i$ into states in $\mathcal{P}_{i-1}$.
If a new state already satisfies some pattern, accumulate its probability, otherwise put it into the set $\mathcal{P}_i$.
When a new item $\sigma_i$ is inserted into $\seq{\alpha, \beta}$ at position $j$, update $\seq{\alpha_\iToj, \beta_\iToj}$ as follows:

\begin{itemize}
    \itemsep -0.3em
    \item $\alpha_\iToj(l) = min(\alpha(l), j)$ if $l \in \lambda(\sigma_i)$ and $l$ is L-type.
    \item $\beta_\iToj(l) = max(\beta(l), j)$ if $l \in \lambda(\sigma_i)$ and $l$ is R-type.
    \item $\alpha_\iToj(l) = \alpha(l) + 1$ if $l \notin \lambda(\sigma_i)$ and $\alpha(l) \geq j$.
    \item $\beta_\iToj(l) = \beta(l) + 1$ if $l \notin \lambda(\sigma_i)$ and $\beta(l) \geq j$.
\end{itemize}

\begin{example}
    Let $\LRIM(\bsigma_0, \Pi_0, \lambda_0)$ denote a labeled RIM where $\bsigma_0 = \ranking{a,b,c,d}$.
	Let $G = g_1 \cup g_2$ be a pattern union where $g_1 = \set{l_1 \succ r_1, l_1 \succ r_2}$, on which we will focus right now.
	Assume that item $a$ and $c$ are associated with label $l_1$, while $b$ with label $r_1$, $d$ with label $r_2$, according to $\lambda_0$.  Below are some solver execution scenarios.
    
    (i) At step 1, item $a$ is inserted at position 1 with probability $\Pi_0(1,1)=1$, thus $\alpha_{1 \rightarrow 1}(l_1)=1$.  (ii)If at step 2, item $b$ is inserted before $a$ with probability $\Pi_0(2,1)$, $\beta_{2 \rightarrow 1}(r_1)=1$ and $\alpha_{1 \rightarrow 1}(l_1) + 1 = 2$.  If item $b$ is inserted after $a$ with probability $\Pi_0(2,2)$, $\beta_{2 \rightarrow 1}(r_1)=2$. Edge $(l_1, r_1)$ is already satisfied by this state, so there is no need to track $r_1$ any more. The $G_u$ will have $g_1 = \set{l_1 \succ r_2}$.  (iii)For the state informally represented by $\set{l_1 \mapsto 2, r_1 \mapsto 1}$, if at step 3, item $c$ is inserted after $b$ at position 2 with probability $\Pi_0(3,2)$ or at position 3 with probability $\Pi_0(3,3)$, edge $(l_1, r_1)$ is violated, which leads to pattern $g_1$ getting violated. The $G_u$ will remove $g_1$ and only track $g_2$ later.  (iv)If at step 4, item $d$ is inserted after $a$ or $c$, pattern $g_1$ is satisfied by the new state, then $G$ is satisfied no matter what the ``situation'' of $g_2$ is, and the probability of this state is accumulated into the marginal probability $prob$.
\end{example}

\begin{theorem}
Given $\LRIM(\bsigma, \Pi, \lambda)$ and a union of bipartite patterns $G$, Algorithm~\ref{alg:bipartite} returns $\Pr(G \mid \bsigma, \Pi, \lambda)$, the marginal probability of $G$ over $\LRIM(\bsigma, \Pi, \lambda)$.
\end{theorem}

\begin{proof}
	Algorithm~\ref{alg:bipartite} is a search algorithm that targets rankings satisfying at least one pattern $g \in G$. 
	Instead of enumerating all $m!$ rankings in the search space, the algorithm runs RIM and inspects the generated rankings of length $i$ at step $i \in [1, m]$. 
	The generated rankings are grouped by their Min/Max positions of labels into states in the form of $\seq{\alpha, \beta}$.
	The algorithm tracks states that can potentially satisfy $G$. 
	Once a state satisfies $G$, its probability will be accumulated into $prob$, and the algorithm will stop tracking it. 
	At step $i$, $\mathcal{P}_i$ is the set of states that can potentially satisfy $G$, $\mathcal{E}_i(\seq{\alpha, \beta})$ maps $\seq{\alpha, \beta}$ to the uncertain patterns and their uncertain edges for this state, and $q_i(\seq{\alpha, \beta})$ is the sum of probabilities of the rankings included in $\seq{\alpha, \beta}$. 
	We prove correctness of Algorithm~\ref{alg:two_label} by induction.
	
	At step 0, there is only one state $\seq{\set{}, \set{}}$ tracking an empty ranking, since no item is inserted yet. 
	All edges in $G$ are uncertain, and the probability of this state is initialized to be 1, which means that no ranking or state is pruned yet. 
	The $prob=0$ since there is also no ranking or state satisfying $G$ yet. 
	
	At step 1, item $\sigma_i$ is inserted into an empty ranking represented by $\seq{\set{}, \set{}}$ at position 1 with probability $\Pi(1,1)=1$.
	State $\seq{\alpha_1, \beta_1}$ is generated, and $q_1(\seq{\alpha_1, \beta_1}) = 1$.
	If $\lambda(\sigma_1) = \emptyset$, $\seq{\alpha_1, \beta_1}  = \seq{\set{}, \set{}}$; otherwise, $\forall l \in \lambda(\sigma_1)$, $\alpha_1(l) = 1$ if $l$ is L-type, and $\beta_1(l) = 1$ if $l$ is R-type.
	Only one ranking $\ranking{\sigma_1}$ is generated at step 1 and all edges in $G$ still remain uncertain.
	So $\mathcal{E}_1(\seq{\alpha_1, \beta_1}) = G$, $prob=0$, and $\mathcal{P}_1 = \set{\seq{\alpha_1, \beta_1}}$ has included all states over $A(\bsigma^1)$ that can potentially satisfy $G$.

	At step $i$, the algorithm reads states from the previous iteration $\mathcal{P}_{i-1}$, as well as $q_{i-1}$ and $\mathcal{E}_{i-1}$. 
	These states are over the first $(i-1)$ items in $\bsigma$, denoted $A(\bsigma^{i-1})$. 
	Assume that the states in $\mathcal{P}_{i-1}$ include all rankings over $A(\bsigma^{i-1})$ that potentially satisfy $G$, that the corresponding probabilities in $q_{i-1}$ and uncertain patterns in $\mathcal{E}_{i-1}$ are correct, and that current $prob$ is the sum of probabilities of all rankings over $A(\bsigma^{i-1})$ that satisfy $G$.
	Note that any ranking $\btau$ over $A(\bsigma^i)$ can always be generated by inserting $\sigma_i$ into $\btau_{-\sigma_i}$ that is a ranking with $\sigma_i$ removed from $\btau$.
	If $\Pr(\btau_{-\sigma_i})$ is already included in $prob$, ranking $\btau$ will keep satisfying $G$ wherever $\sigma_i$ is inserted.
	If $\btau_{-\sigma_i}$ already violates $G$ at step $(i-1)$, ranking $\btau$ will keep violating $G$ wherever $\sigma_i$ is inserted.
	So any ranking $\btau$ included by states in $\mathcal{P}_i$ must be generated from $\btau_{-\sigma_i}$ included by states in $\mathcal{P}_{i-1}$.
	If a new generated state satisfies $G$, it must also be generated from a state in $\mathcal{P}_{i-1}$.
	Inserting $\sigma_i$ into every state of $\mathcal{P}_{i-1}$ at every possible position $j$ will generate all states required by $\mathcal{P}_i$ and the incremental part of $prob$.

	Let $\seq{\alpha_\iToj, \beta_\iToj}$ denote a new state by inserting $\sigma_i$ into $\seq{\alpha, \beta} \in \mathcal{P}_{i-1}$ at position $j$. 
	The values of $\alpha_\iToj$ and $\beta_\iToj$ should be updated according to the algorithm description in order to reflect the Min/Max positions correctly, so that the algorithm can determine whether $\seq{\alpha_\iToj, \beta_\iToj}$ satisfies $G$.
	The state $\seq{\alpha_\iToj, \beta_\iToj}$ falls into one of the following 3 cases:
    \begin{itemize}
    \itemsep -0.3em
    \item Case 1: $\seq{\alpha_\iToj, \beta_\iToj}$ violates all patterns in $G$. The algorithm prunes this state.
    \item Case 2: 
    			$\seq{\alpha_\iToj, \beta_\iToj}$ satisfies a pattern $g \in G$. 
    			Its probability $\Pr(\seq{\alpha_\iToj, \beta_\iToj})$ is accumulated into $prob$ and the algorithm stops tracking this state. 
    			Recall that $q_{i-1}(\seq{\alpha, \beta})$ is the sum of the probabilities of rankings $\set{\lst[N]{\btau}}$ included by it.
    			Then $\Pr(\seq{\alpha_\iToj, \beta_\iToj}) = \sum_{k=1}^{N} {\big( \Pr(\btau_k) \cdot \Pi(i, j) \big)} = q_{i-1}(\seq{\alpha, \beta}) \cdot \Pi(i, j)$.
    \item Case 3: 
    			$\seq{\alpha_\iToj, \beta_\iToj}$ can still potentially satisfy $G$ in the future, so it is added into $\mathcal{P}_i$.
    			Its probability is calculated the same way as above: $\Pr(\seq{\alpha_\iToj, \beta_\iToj}) = q_{i-1}(\seq{\alpha, \beta}) \cdot \Pi(i, j)$, and tracked by $q_i$.
    			Calculate the uncertain part of $G$ for this state, $\mathcal{E}_i(\seq{\alpha_\iToj, \beta_\iToj})$, with the latest Min/Max positions of labels.
    \end{itemize}
 	Then, after iterating over all states in $\mathcal{P}_{i-1}$ and all positions $j \in \set{1, \ldots, i}$, the states in $\mathcal{P}_{i}$ include all rankings over $A(\bsigma^i)$ that potentially satisfy $G$, and $prob$ is the sum of probabilities of all rankings over $A(\bsigma^i)$ having satisfied $G$.
	The probabilities in $q_i$ and the uncertain patterns in $\mathcal{E}_i$ are also updated correctly.

	At step $m$, all items are inserted, there remain no uncertain states in $\mathcal{P}_m$, and $prob$ includes the probability of all rankings that satisfy some pattern in $G$.
\end{proof}

\subsubsection{Bipartite Solver for Upper Bounds}
\label{sec:approx:ub}

Let $tc(g)$ be the transitive closure of pattern $g$.
Each edge $(l, r) \in tc(g)$ represents a constraint $\alpha(l) < \beta(r)$.
Let $U$ denote the set of these constraints.
By the definition of \e{label embedding}, any ranking $\btau$ satisfying $g$ must satisfy $U$, denoted by $\btau \models U$, so $U$ gives an upper bound of $g$.
\begin{example}
    Let $g_0 = \set{l_a \succ l_b, l_b \succ l_c}$, a linear order $l_a \succ l_b \succ l_c$.
    Then $tc(g_0) = \set{l_a \succ l_b, l_b \succ l_c, l_a \succ l_c}$, and $U_0 = \set{\alpha(l_a) < \beta(l_b), \alpha(l_b) < \beta(l_c), \alpha(l_a) < \beta(l_c)}$ accordingly. 
    If a ranking $\btau_0 \models g_0$, $\btau_0$ must satisfy all constraints in $U_0$. 
    But if $\btau_0 \models U_0$, it is possible that $\btau_0 \not\models g_0$. 
    For example $\btau_0 = \ranking{b_1, a, c, b_2}$ w.r.t. $\lambda_0 = \set{a \mapsto \set{l_a}, b_1 \mapsto \set{l_b}, b_2 \mapsto \set{l_b}, c \mapsto \set{l_c}}$. 
    In this case, $\btau_0 \models U_0$ but $\btau_0 \not\models g_0$.
\end{example}

For a pattern union $G = g_1 \cup \ldots \cup g_\numPatterns$, we can also calculate its upper bound in a similar way. Let $U_i$ denote the upper bound constraints for $g_i \in G$. For any ranking $\btau$, $\btau \models G$ \ifff $\exists g_i \in G, \btau \models g_i$. The $U_i$ is less strict than $g_i$, so $\btau \models U_i$ if $\btau \models g_i$. Let $\mathcal{U} = U_1 \cup \ldots \cup U_\numPatterns$ denote the union of upper bound constraints. Then $\btau \models \mathcal{U}$ iff $\exists U_i \in \mathcal{U}, \btau \models U_i$. So $\btau \models \mathcal{U}$ if $\btau \models G$. The $\mathcal{U}$ gives an upper bound for $G$.

Let $U_s$ denote a subset of $U$. Note that $U_s$ also gives an upper bound of $g$ that is less strict than the original $U$, but is faster to calculate. The same conclusion applies to a union of constraint subsets. This is the principle behind the evaluation of Most-Probable-Session queries in Section~\ref{sec:overview:count}.
\section{Approximate solvers}
\label{sec:approx}

Exact solvers compute answers to intractable problems. We will study their performance empirically in Section~\ref{sec:exp:exact}, and will observe that these solves are practical only for small queries, and for a modest number of candidates. To address scalability challenges that are inherent in the problem, we design approximate solvers that leverage the structure of the Mallows model, and specifically the recent results on efficient sampling from the Mallows posterior~\cite{DBLP:journals/jmlr/LuB14}. 

Let $\LMAL(\bsigma, \phi, \lambda)$ denote a labeled Mallows model with labeling function $\lambda$.
Let $G = g_1 \cup \ldots \cup g_\numPatterns$ be a union of $\numPatterns$ patterns.
We are interested in $\Pr(G \mid \bsigma, \phi, \lambda)$, the marginal probability of $G$ over $\LMAL(\bsigma, \phi, \lambda)$.
This is also the posterior probability of $G$ over $\LMAL(\bsigma, \phi, \lambda)$, or the expectation that a sample $\btau$ from $\mallows(\bsigma, \phi)$ satisfies $G$ w.r.t. $\lambda$.
\[
\Pr(G \mid \bsigma, \phi, \lambda) = \mathds{E}\Big(\mathds{1}\big((\btau,\lambda) \models G \big)\Big), \btau \sim \mallows(\bsigma, \phi)
\]
where $\mathds{1}(x)$ is the indicator function.

\subsection{Importance Sampling for Mallows}

Sampling is popular for probability estimation.
For example, we can use Rejection Sampling (RS) to sample a large number of rankings from $\mallows(\bsigma, \phi)$ and count how many of them satisfy $G$.
Generally, RS works well if the target probability is high, but is impractical for estimating low-probability events.
Importance Sampling (IS) can effectively estimate rare events~\cite{kahn1950random1, kahn1950random2}.
IS estimates the expected value of a function $f(x)$ in a probability space $\oriDistr$ via sampling from another {\em proposal distribution} $\propDistr$, then re-weights the samples for unbiased estimation.
Assume that $x$ is discrete, and that $N$ samples $\{ x_1, x_2, ..., x_N \}$ are generated from $\propDistr$.
The estimation is done as follows:

\begin{equation} \label{eq:IS}
\begin{split}
\mathds{E}_{\oriDistr}\big( f(x) \big) 
& = \sum_{x \in \oriDistr} {f(x) \cdot p(x)} 
= \sum_{x \in \propDistr} {f(x) \cdot \frac{p(x)}{q(x)} \cdot q(x)} \\
& = \mathds{E}_{\propDistr} \left( f(x) \cdot \frac{p(x)}{q(x)} \right)  
\approx \frac{1}{N} \sum_{i=1}^{N} {\frac{p(x_i)}{q(x_i)} f(x_i)}
\end{split}
\end{equation}
where $p(x)=\Pr(x \mid \oriDistr)$ and $q(x) = \Pr(x \mid \propDistr)$.

IS re-weights each sample $x_i$ by an \e{importance factor} $\frac{p(x_i)}{q(x_i)}$. 
When applying IS, $\propDistr$ is chosen to support efficient sampling and, ideally, to provide estimates $q(x)$ that are close to $p(x)$, also for efficiency reasons.
To calculate $\Pr(G | \bsigma, \phi, \lambda) = \mathds{E}(\mathds{1}((\btau,\lambda) \models G ))$, we set $f(x) = \mathds{1}((\btau,\lambda) \models G )$, where ranking $\btau$ is a sample.

\subsection{From Pattern Union to Sub-ranking Union}

Before diving into details of applying IS to RIM inference, let us examine the meaning of $(\btau,\lambda) \models G$.
Previously, we had $(\btau,\lambda) \models G$ \ifff $\exists g \in G, (\btau,\lambda) \models g$.
Recall from Section~\ref{sec:preliminaries:pref} that $(\btau,\lambda) \models g$ if there exists an embedding function $\embedding$ in which labels match ($\forall l\in nodes(g), l \in \lambda(\btau(\embedding(l)))$) and edges match ($\forall (l, l') \in edges(g), \embedding(l)<\embedding(l')$).

The embedding $\embedding$ constructs a partial order $\partialOrder = \{ \embedding(l) \succ \embedding(l') | (l, l') \in edges(g) \}$ so that $\btau \in \Omega(\partialOrder)$.
(Recall that $\Omega(\partialOrder)$ is the set of linear extensions of $\partialOrder$.)
Conceptually, a pattern $g$ can be decomposed into a union of partial orders with different embedding functions.
Let $\extension(g, \lambda)$ denote the union of partial orders decomposed from $g$ w.r.t. $\lambda$.
Then $(\btau, \lambda) \models g$ \ifff $\exists \partialOrder \in \extension(g, \lambda), \btau \in \Omega(\partialOrder)$.
We can calculate these partial orders for all patterns in $G$, any permutation $\btau$ satisfying any partial order will immediately satisfy a pattern in $G$, and so will satisfy $G$ itself. In this sense, $G$ is equivalent to a union of partial orders.

A partial order $\partialOrder$ can further be decomposed into a union of sub-rankings that are consistent with $\partialOrder$.
For example, $\partialOrder = \{a \succ c, b \succ c\}$ has two sub-rankings $\subRanking_1 = \angs{a, b, c}$ and $\subRanking_2 = \angs{b, a, c}$.
Let $\extension(\partialOrder)$ denote the union of sub-rankings from partial order $\partialOrder$. Let $\btau \models \subRanking$ denote that a permutation $\btau$ is consistent with a sub-ranking $\subRanking$.
Then $\btau \in \Omega(\partialOrder)$ \ifff $\exists \subRanking \in \extension(\partialOrder), \btau \models \subRanking$.
Because $G$ is equivalent to a union of partial orders (w.r.t. $\lambda$), we have:
$G = \bigcup \{ \subRanking | \subRanking \in \extension(\partialOrder), \partialOrder \in \extension(g, \lambda), g \in G \}$.

Figure~\ref{fig:decompose_patterns_to_subRankings} is an example, where a union of two patterns is decomposed into three partial orders, then further into six sub-rankings.
A ranking satisfies the pattern union \ifff it satisfies the sub-ranking union.
Assuming $\numPatterns$ patterns are decomposed into $\numSubRankings$ sub-rankings, we have $G=g_1 \cup \ldots \cup g_\numPatterns = \subRanking_1 \cup \ldots \cup \subRanking_\numSubRankings$.

\begin{figure}[htb]
	\centering
	\includegraphics[width=0.9\linewidth]{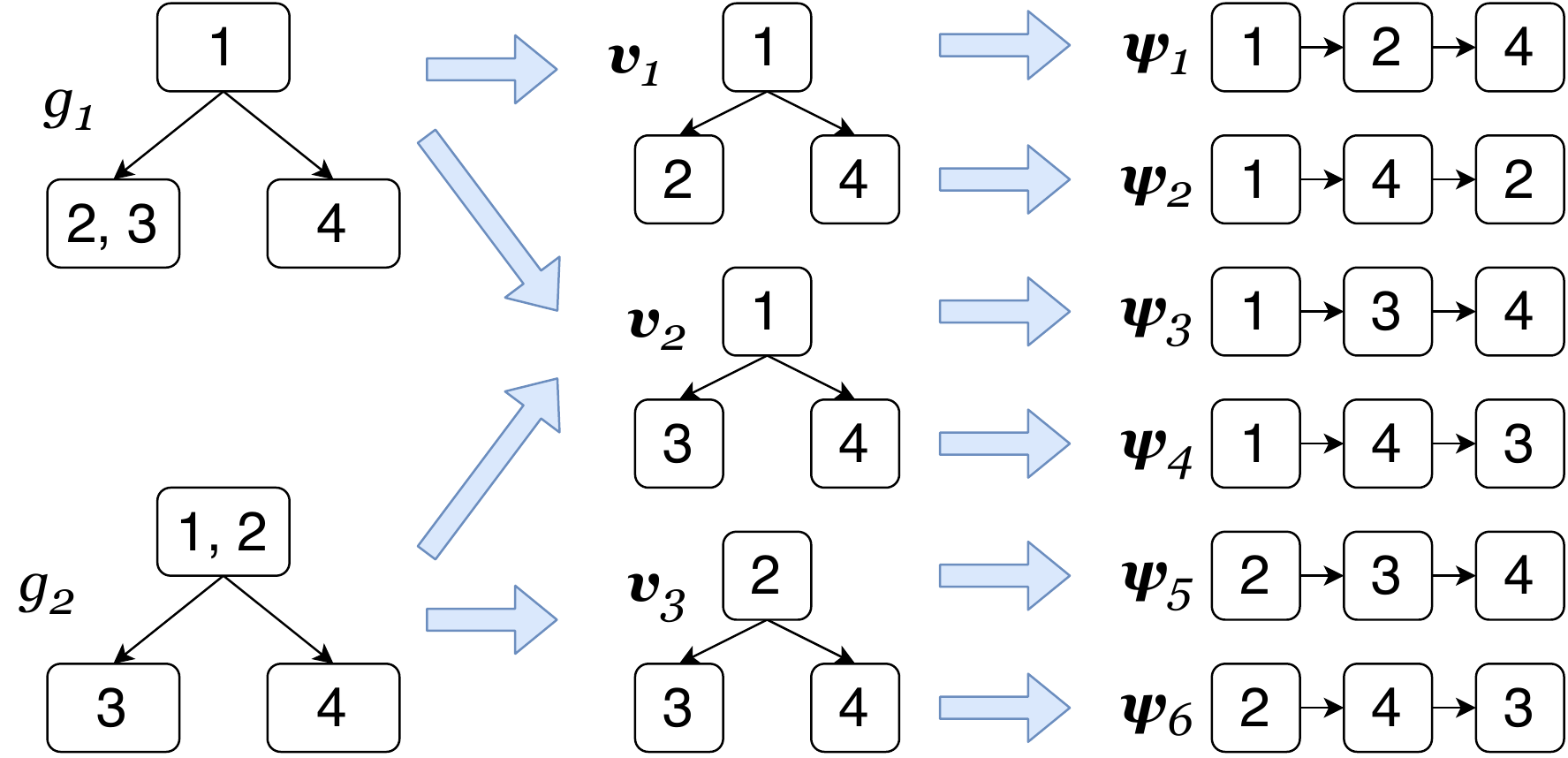}
	\caption{A union of two patterns (left), decomposed into a union of three partial orders (middle), then a union of six sub-rankings (right).}
	\label{fig:decompose_patterns_to_subRankings}
\end{figure}

\subsection{IS-AMP for a Single Sub-ranking}
\label{sec:approx:is}

The pattern union $G$ has been decomposed into $\numSubRankings$ sub-rankings. 
Before dealing with all sub-rankings, let us see how to estimate the expectation of a single sub-ranking $\subRanking$ over the Mallows model $\mallows(\bsigma, \phi)$.

If $\Pr(\subRanking \mid \btau, \phi)$ is low, RS is inefficient to reach accurate estimation. 
We can apply IS instead, using a proposal distribution that easily generates permutations satisfying $\subRanking$. 

Our first method, called IS-AMP, uses AMP, a state-of-the-art Mallows sampler conditioned on a partial order of items~\cite{DBLP:journals/jmlr/LuB14}, to construct a proposal distribution.
IS-AMP works well when the proposal distribution is around the ``important region'' of the probability space.
Unfortunately, as we  show next, AMP does not always give desirable proposal distributions, especially when there are  multiple \e{modals} --- peaks or local maxima --- in the posterior distribution. 

\begin{example}
\label{example: isamp_fails}

Let  $\subRanking_0=\ranking{\sigma_3, \sigma_1}$ be a sub-ranking for which we wish to calculate the expectation over $\mallows(\bsigma_0, \phi_0)$, with $\bsigma_0=\ranking{\sigma_1, \sigma_2, \sigma_3}$ and $\phi_0=0.01$.
Recall that with $\phi_0=0.01$, much of the probability mass of $\mallows(\bsigma_0, \phi_0)$ is around $\bsigma_0$.
In this case IS-AMP will sample $\btau_0 = \ranking{\sigma_3, \sigma_1, \sigma_2}$ very frequently, as follows:
(i) Insert $\sigma_1$ into an empty ranking $\ranking{}$.
(ii) Insert $\sigma_2$ into $\ranking{\sigma_1}$ after $\sigma_1$ with probability $\frac{1}{1+0.01}$.
(iii) Insert $\sigma_3$ into $\ranking{\sigma_1, \sigma_2}$ before $\sigma_1$ with probability 1.
If all samples are $\btau_0 = \ranking{\sigma_3, \sigma_1, \sigma_2}$, we will estimate:
\begin{equation*}
\begin{split}
\text{IS-AMP}(\subRanking_0 \mid \mallows(\bsigma_0, \phi_0)) 
& \approx \frac{\Pr(\btau_0 \mid \mallows(\bsigma_0, \phi_0))}{\Pr(\btau_0 \mid \AMP(\bsigma_0, \phi_0, \subRanking_0))} \\
& \approx \frac{9.9 \times 10^{-5}}{0.99} = 10^{-4}
\end{split}
\end{equation*}

However, there are two \e{modals} in the posterior distribution, $\btau_0 = \ranking{\sigma_3, \sigma_1, \sigma_2}$ and $\btau_1 = \ranking{\sigma_2, \sigma_3, \sigma_1}$.
These modals are rankings that are closest to $\bsigma_0$ (in terms of Kendall-tau distance) among those that are consistent with $\btau_0$, and so much of the probability mass of the posterior distribution is concentrated around them, {\em not} around $\bsigma_0$.
We have:
\begin{equation*}
\begin{split}
\Pr(\subRanking_0 \mid \bsigma_0,\phi_0) 
& \geq \Pr(\btau_0 \mid \bsigma_0, \phi_0)  + \Pr(\btau_1 \mid \bsigma_0, \phi_0) \\
& \approx 10^{-4} + 10^{-4}  > \text{IS-AMP}(\subRanking_0 \mid \bsigma_0, \phi_0)
\end{split}
\end{equation*}
\end{example}

In the example above, IS-AMP fails to effectively estimate the probability, because the posterior distribution is multi-modal.  To address this issue, we design MIS-AMP, a new sampler based on AMP geared specifically at multi-modal distributions.   We describe MIS-AMP next.  

\subsection{MIS-AMP for a Single Sub-ranking}
\label{sec:approx:mis}

We first give some general background on Multiple Importance Sampling (MIS), and will then show how it is applied to our scenario.
Assume there are $\numPropDistrs$ proposal distributions with probability mass functions $\{ q_1, ..., q_\numPropDistrs \}$ and $n_i$ samples generated from $q_i$.
Let $x_{i, j}$ be the $j$-th sample generated from $q_i$.
For each $x_{i, j}$, MIS not only calculates its importance factor as does IS, but it also calculates a weight $w_i$ with which $x_{i,j}$ is sampled from $q_i$.
Let $N = \sum_{i=1}^\numPropDistrs {n_i}$ and $c_i = n_i / N$.
Let $f(x)$ be the function of which we want to compute the expectation, and $p(x)$ be the probability mass function of the original distribution.
The MIS estimator is:
\begin{equation} \label{eq:MIS}
\mathds{E}(f(x)) = \sum_{i=1}^\numPropDistrs {\frac{1}{n_i} \sum_{j=1}^{n_i} {w_i(x_{i,j})\frac{p(x_{i,j})}{q_i(x_{i,j})}f(x_{i,j})}}
\end{equation}
This estimator is unbiased if $\forall x, \sum_i {w_i(x)} = 1$.
Vech and Guibas~\cite{DBLP:conf/siggraph/VeachG95} showed that the weighting function $w_i(x) = \frac{c_i q_i(x)}{\sum_{t=1}^\numPropDistrs c_t q_t(x)}$, designed to balance the contribution of each proposal distribution to the estimate, is a good choice.

When generating an equal number of samples from all proposal distributions (i.e., $n_1=...=n_\numPropDistrs=n$ and $c_1=...=c_\numPropDistrs=1/\numPropDistrs$), the Equation~\eqref{eq:MIS} can be simplified as:
\begin{equation} \label{eq:MIS_simplified}
\mathds{E}(f(x)) = \frac{1}{\numPropDistrs \cdot n} \sum_{i=1}^\numPropDistrs {\sum_{j=1}^{n} {\frac{p(x_{i,j})}{\frac{1}{\numPropDistrs}\sum_{t=1}^\numPropDistrs q_t(x_{i,j})}f(x_{i,j})}}
\end{equation}

\paragraph*{Importance Sampling for Mallows}
A good proposal distribution for IS should produce more samples in the ``important region'' of the target distribution---the region 
wherein there is a significant probability mass. 
So, instead of sampling with the original Mallows, we sample permutations that are consistent with the sub-ranking $\subRanking$.
Among all such, the ones that are nearest to Mallows center $\bsigma$ are the modals of the posterior.
The samples around these modals are the important regions, and they should be effectively captured by the proposal distributions.

 Our strategy is to construct Mallows models centered at these modals, and run AMP over them conditioned on the sub-ranking $\subRanking$.
 Unfortunately, it is intractable to find a completion of a partial order that is closest, in terms of Kendall-tau distance, to a given ranking (Theorem 2 in~\cite{DBLP:conf/walcom/BrandenburgGH12}).  This makes finding the  modals consistent with $\subRanking$ that are closest to $\bsigma$ intractable.
 Algorithm~\ref{alg:find_modals} uses a greedy heuristic to search for modals, by inserting items into $\subRanking$ at positions that minimize the distance to $\bsigma$.
 Note that $\subRanking_\iToj$ is a sub-ranking, with $\sigma_i$ inserted into $\subRanking$ at position $j$.

Let $S=\set{\lst[\numPropDistrs]{\bsigma}}$ denote the set of modals output by Algorithm~\ref{alg:find_modals}.
We construct $\mallows(\bsigma_1, \phi), \ldots, \mallows(\bsigma_\numPropDistrs, \phi)$, and run AMP over each, conditioned on the sub-ranking $\subRanking$, raising $\numPropDistrs$ proposal distributions.
We are interested in the expectation of $\mathds{1}(\btau \models \subRanking)$, where $\btau \sim \mallows(\bsigma, \phi)$. Note that the permutations generated by MIS-AMP will always satisfy $\subRanking$, \ie $\mathds{1}(\btau \models \subRanking) \equiv 1$.
Using Equation~\eqref{eq:MIS_simplified}, we estimate:
\begin{equation} \label{eq:MIS_AMP}
\mathds{E} \big( \mathds{1}(\btau \models \subRanking) \big) = \frac{1}{\numPropDistrs \cdot n} \sum_{i=1}^\numPropDistrs {\sum_{j=1}^{n} {\frac{p(x_{i,j})}{\frac{1}{\numPropDistrs}\sum_{t=1}^\numPropDistrs q_t(x_{i,j})}}}
\end{equation}

\begin{example}
We now revisit Example~\ref{example: isamp_fails} and solve it by MIS-AMP.
Recall that we wish to calculate the expectation of $\subRanking_0=\ranking{\sigma_3, \sigma_1}$ over $\mallows(\bsigma_0, \phi_0)$ with $\bsigma_0=\ranking{\sigma_1, \sigma_2, \sigma_3}$ and $\phi_0=0.01$.  Algorithm~\ref{alg:find_modals} will find two modals, $\bsigma_1 = \ranking{\sigma_3, \sigma_1, \sigma_2}$ and $\bsigma_2 = \ranking{\sigma_2, \sigma_3, \sigma_1}$ as centers of the newly constructed  $\mallows(\bsigma_1, \phi)$ and $\mallows(\bsigma_2, \phi)$.
MIS-AMP then draws rankings from two AMP samplers, $\AMP(\bsigma_1, \phi_0, \subRanking_0)$ and $\AMP(\bsigma_2, \phi_0, \subRanking_0)$. Then $\btau_0 = \ranking{\sigma_3, \sigma_1, \sigma_2}$ is re-weighted as follows.
\begin{equation*}
\begin{split}
& \text{MIS{-}AMP}(\btau_0 \mid \bsigma_0, \phi_0, \subRanking_0) \\
\approx & \frac{\Pr(\btau_0 \mid \mallows(\bsigma_0, \phi_0))}{\frac{1}{2}\Big(\Pr\big(\btau_0 | \AMP(\bsigma_1, \phi_0, \subRanking_0)\big){+}\Pr\big(\btau_0 | \AMP(\bsigma_2, \phi_0, \subRanking_0)\big)\Big)} \\
\approx & \frac{9.9 \times 10^{-5}}{\frac{1}{2}(0.99 + 0.01)} \approx 2 \times 10^{-4}
\end{split}
\end{equation*}

That is, in terms of re-weighting $\btau_0$, MIS-AMP significantly outperforms IS-AMP in Example~\ref{example: isamp_fails}.
\end{example}

Having discussed how MIS-AMP can be used to estimate the posterior probability for a single sub-ranking, we now return to the more general problem we study in this paper, and show how MIS can be used to estimate the probability of a union of sub-rankings and a union of patterns.

\begin{algorithm}[t!]
	\caption{GreedyModals}
	\begin{algorithmic}[1]
		\REQUIRE Sub-ranking $\subRanking$, Mallows model $\mallows(\bsigma, \phi)$
		\STATE $S \defeq \{ \subRanking \}$
		\FOR {$i=1,2,...,m$}
			\IF {$\sigma_i \notin \subRanking$}
				\STATE $S' \defeq \emptyset$
				\FOR {$\subRanking \in S$}
					\STATE $J = \{j \mid \dist(\subRanking_\iToj, \bsigma) =   \underset{j'=1,\dots |\subRanking|}{\min} \ {\hskip-1em\dist(\subRanking_{i \rightarrow j'}, \bsigma)}\}$ 
					\FOR {$j \in J$}
						\STATE $S'.add(\subRanking_\iToj)$.
					\ENDFOR					
				\ENDFOR
				\STATE $S \defeq S'$
			\ENDIF
		\ENDFOR
		\RETURN  $S$
	\end{algorithmic}
	\label{alg:find_modals}
\end{algorithm}

 \begin{algorithm}[t!]
 	\caption{ApproximateDistance}
 	\begin{algorithmic}[1]
 		\REQUIRE Sub-ranking $\subRanking$, Mallows center $\bsigma$
 		\STATE $\btau \defeq \subRanking$
 		\FOR {$i=1,2,...,m$}
 		\IF {$\sigma_i \notin \subRanking$}
 		\STATE $J = \{j \mid \dist(\subRanking_\iToj, \bsigma) =   \underset{j'=1,\dots |\subRanking|}{\min} \ {\hskip-0em\dist(\subRanking_{i \rightarrow j'}, \bsigma)}\}$
 		\STATE $\btau \defeq \btau_\iToj, j \in J$
 		\ENDIF
 		\ENDFOR
 		\RETURN  Kendall-tau$(\btau, \bsigma)$
 	\end{algorithmic}
 	\label{alg:approx_dist_subranking}
 \end{algorithm}

\subsection{MIS-AMP-Lite and MIS-AMP-Adaptive}
\label{sec:approx:lite}

MIS-AMP can in principle be used for a union of sub-rankings and a union of patterns.  However, not unexpectedly, the challenge is that a pattern union $G$ corresoponds to exponentially many sub-rankings, each of which in turn  yields multiple modals for MIS (per Section~\ref{sec:approx:mis}), and so generating all sub-rankings and then using MIS-AMP for each is intractable.  Instead, we develop a method for selecting a subset of subrankings of fixed size $\numPropDistrs$, and ensuring that the corresponding proposal distributions cover the important regions of the posterior.  We call this method MIS-AMP-lite.

Suppose that $G$ has $\numPatterns$ patterns, and that it is equivalent to a union of $\numSubRankings$ sub-rankings.  MIS-AMP-lite 
sorts $\numSubRankings$ sub-rankings in ascending order of their {\em estimated distance} from the Mallows center $\bsigma$, as computed by Algorithm~\ref{alg:approx_dist_subranking}. Since the sub-rankings containing modals close to $\bsigma$ are desirable, we define the distance between a sub-ranking $\subRanking$ and $\bsigma$ as the minimum Kendall-tau distance between $\bsigma$ and a modal contained in $\subRanking$. But identifying the closest modals is intractable, thus we estimate this distance using a greedy modal $r$ generated in Algorithm~\ref{alg:approx_dist_subranking}. Let $\dist(\subRanking, \bsigma)$ denote the estimated Kendall-tau distance between $\subRanking$ and $\bsigma$. Each sub-ranking $\subRanking$ represents a component of size proportional to $\phi^{\dist(\subRanking, \bsigma)}$ in the posterior distribution. 

Since MIS-AMP-lite prunes many components in the posterior distribution, the algorithm should compensate for this pruning in the final result. Let $S$ denote the sub-rankings in $G$, and $S^+ \subseteq S$ denote the set of selected sub-rankings. The compensation factor $\compensation_{\subRanking}$ for sub-ranking pruning is:
\[
\compensation_{\subRanking} = \frac{\sum_{\subRanking \in S} \phi^{\dist(\subRanking,\bsigma)}}{\sum_{\subRanking \in S^+} \phi^{\dist(\subRanking,\bsigma)}}
\]
Intuitively, the compensation factor $\compensation_{\subRanking}$ captures the portion of the probability space represented by the selected sub-rankings.
MIS-AMP-lite also prunes modals, selecting $\numPropDistrs$ modals closest to $\bsigma$. Let $M$ denote the set of available modals, and $M^+ \subseteq M$ denote the set of selected modals.
The compensation factor $\compensation_{r}$ for modal pruning is defined similarly as for sub-rankings:
\[
\compensation_{r} = \frac{\sum_{r \in M} \phi^{\dist(r,\bsigma)}}{\sum_{r \in M^+} \phi^{\dist(r,\bsigma)}}
\]

Let $p$ denote the estimate by MIS-AMP-lite over $\numPropDistrs$ proposal distributions without compensation. The final estimate is
$\Pr(G|\bsigma, \phi) = p \cdot \compensation_{\subRanking} \cdot \compensation_{r}$.
We experimentally validate the compensation mechanism in Section~\ref{sec:exp:approx}, and show that it leads to higher accuracy.

MIS-AMP-lite requires $\numPropDistrs$, the number of proposal distributions, as an input parameter. As an alternative, MIS-AMP-adaptive calls MIS-AMP-lite as a subroutine, and gradually increases the number of proposal distributions in increments of 
 $\Delta \numPropDistrs$ until convergence.  We will demonstrate the effectiveness of MIS-AMP-adaptive in Section~\ref{sec:exp:approx}.

\newpage 
\section{Experimental evaluation}
\label{sec:experiments}

We now present results of an extensive experimental evaluation of exact and approximate solvers over six families of experimental datasets.  All solvers are implemented in Python.  The general solver uses LTM~\cite{DBLP:conf/sigmod/CohenKPKS18}, implemented in Java, as a subroutine. We ran experiments on a 64-bit Ubuntu Linux machine with 48 cores on 4 CPUs of Intel(R) Xeon(R) CPU E5-2680 v3 @ 2.50GHz, and 512GB of RAM.

\begin{figure*} [t!]
	\centering
	\begin{minipage}[t]{.32\textwidth}
		\centering
		\includegraphics[width=\linewidth]{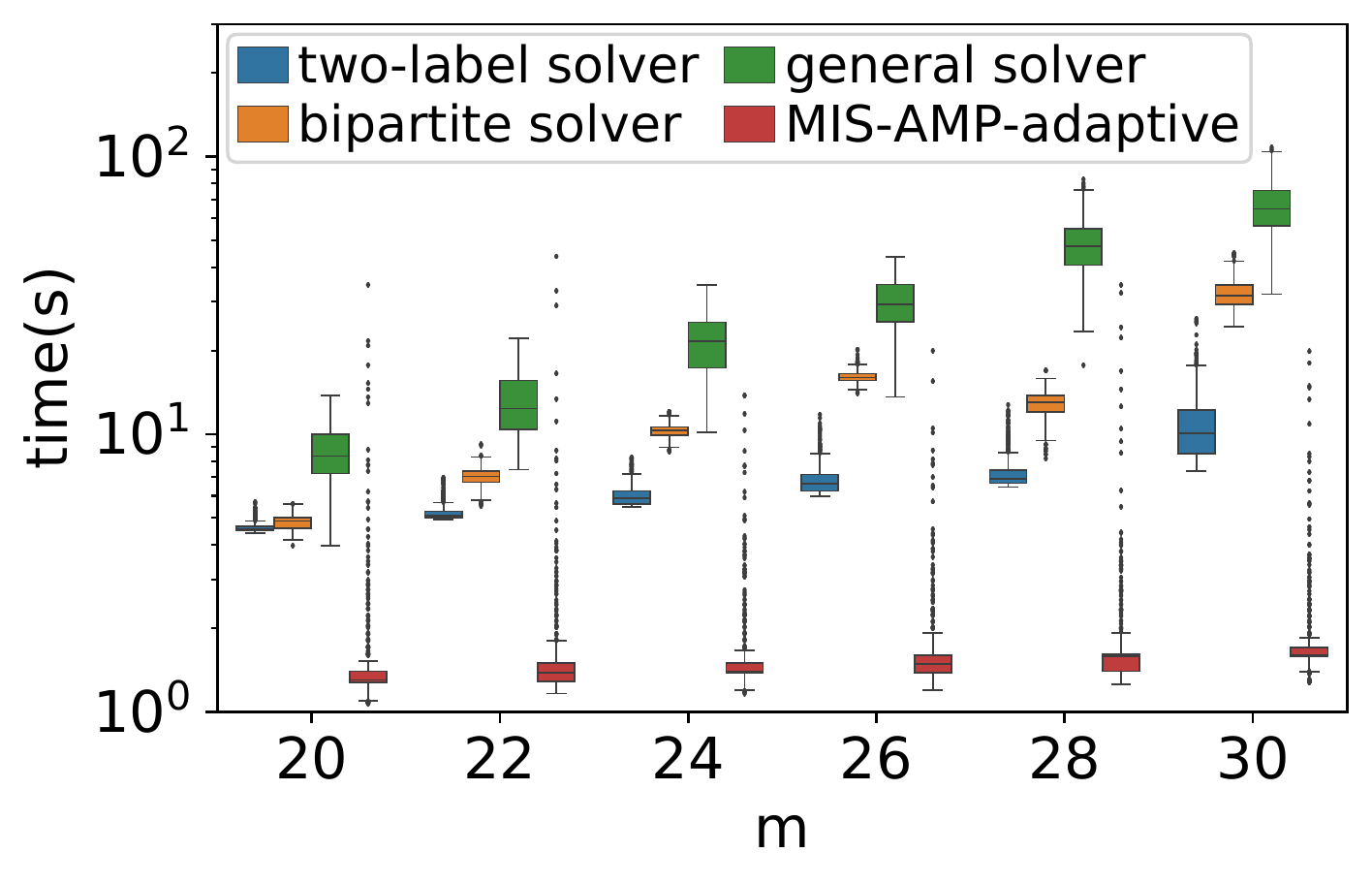}
		\caption{Evaluating a two-label query over \textbf{Polls} to compare performance of exact solvers and of  MIS-AMP-adaptive.}
		\label{fig:runtime__2label_vs_bipartite_vs_inexclu}
	\end{minipage}
	\hfill
	\begin{minipage}[t]{.32\textwidth}
		\centering
		\includegraphics[width=\linewidth]{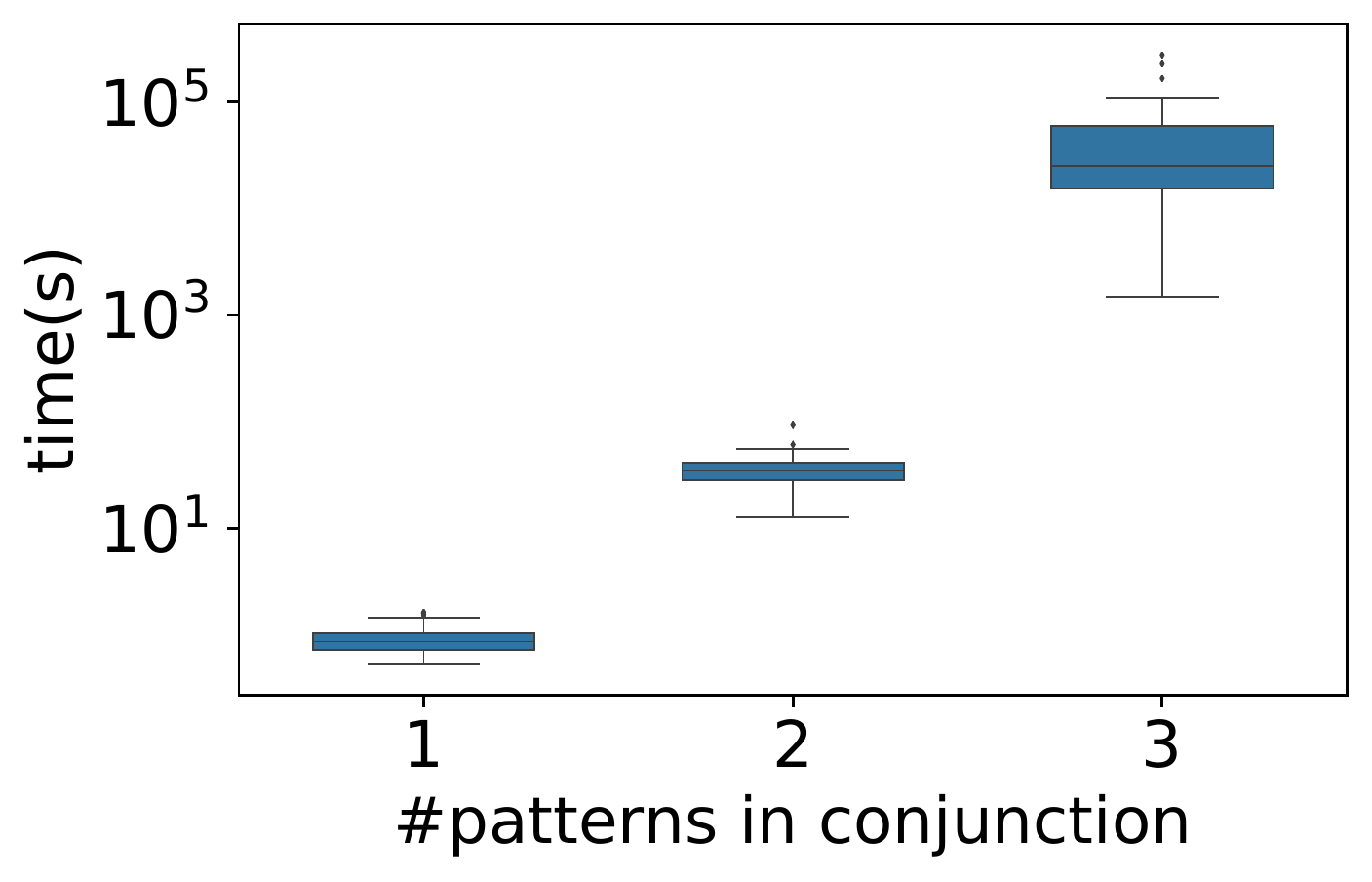}
		\caption{General solver running time increases exponentially with \# patterns in conjunction for \benchmarkA.}
		\label{fig:LTM_runtime_vs_numPatterns}
	\end{minipage}
	\hfill
	\begin{minipage}[t]{.32\textwidth}
		\centering
		\includegraphics[width=0.73\linewidth]{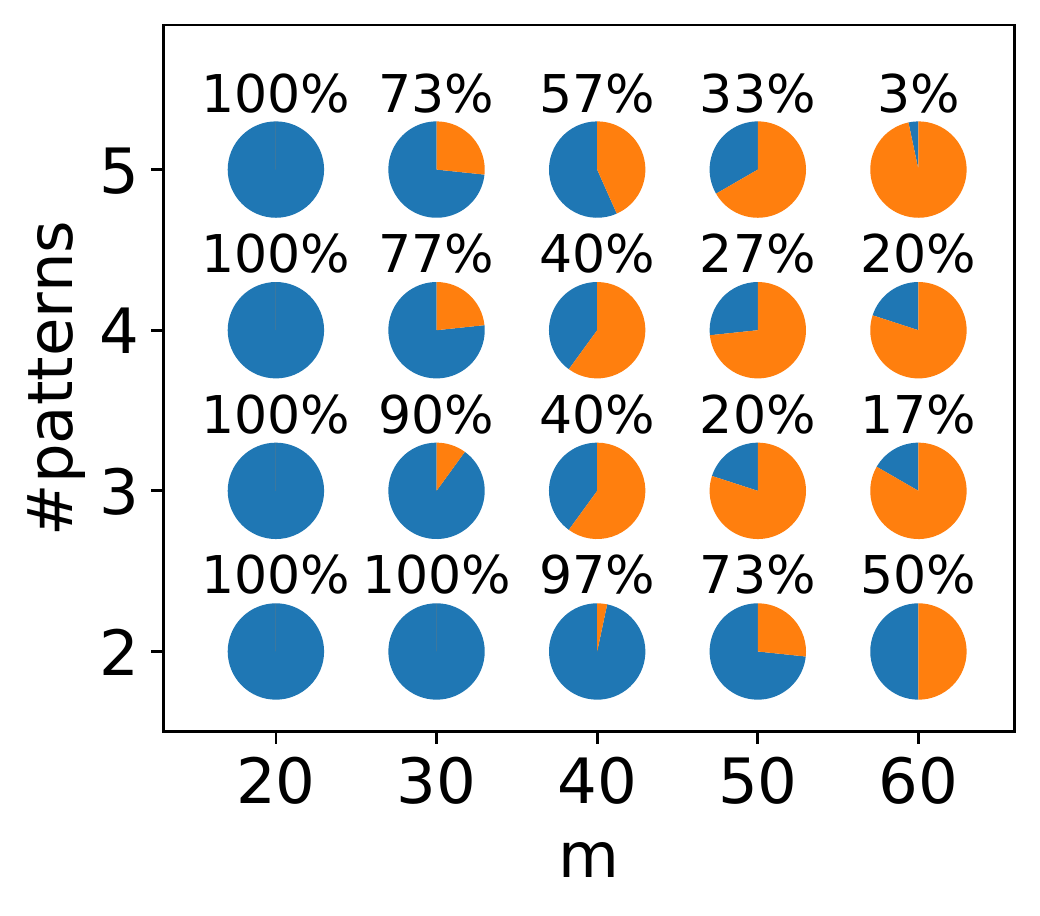}
		\caption{Proportions of instances that finished in 10 minutes by two-label sover over \benchmarkD.}
		\label{fig:scalability_2label}
	\end{minipage}
	\hfill
\end{figure*}

\subsection{Datasets}

In our experimental evaluation we use two real datasets --- \textbf{MovieLens} and \textbf{CrowdRank}, and four synthetic benchmarks --- \textbf{Polls}, and \textbf{Benchmarks A}, \textbf{B}, and \textbf{C}.

\benchmarkA \ has 33 pattern unions over the model $\mallows(\ranking{\lst{\sigma}}, 0.1)$.
Each union consists of 3 bipartite patterns $\{A \succ C, A \succ D, B \succ D \}$.
In every pattern union, the 3 patterns share the same items in label $B$ and $D$. The labels all have 3 items sampled from $\bsigma$. Label $A$ and $B$ get item $\sigma_i$ with probability $p_i \propto i^{1.5}$, while label $C$ and $D$ get item $\sigma_i$ with probability $p_i \propto (16-i)^{1.5}$. Note that items with labels $C$ and $D$ tend to have higher ranks than items with $A$ and $B$. As a result, some pattern unions have low probabilities, allowing us to test the accuracy of approximate solvers.

\benchmarkB \ is a set of pattern unions with varying number of patterns, labels per pattern, and items per label. Within a pattern union, all patterns share the same edges that correspond to random partial order of labels. The number of items $m$ is among $\{20, 50, 100, 200\}$, and Mallows $\phi=0.1$. The number of patterns per union is 1, 2, or 3. The number of labels per pattern is 3, 4, or 5. The number of items per label is 3, 5, or 7. Each combination of the parameters above has 10 instances in this benchmark, for a total of $4 \times 3 \times 3 \times 3 \times 10 = 1080$ instances. This benchmark tests the scalability of approximate solvers.

\benchmarkC \ is a set of bipartite pattern unions with varying number of patterns, labels per pattern, and items per label. The patterns within the same union share the same edges that are random bipartite directed graphs of labels. The number of items $m$ is among $\{10, 12, 14, 16\}$, and Mallows $\phi=0.1$. The number of patterns per union is 1, 2, or 3. The number of labels per pattern is among 2, 3, or 4. The number of items per label is 1, 3, or 5. Each combination of the parameters above has 10 instances, for a total of $1080$ instances. This benchmark has smaller patterns and fewer items in the Mallows models compared to \benchmarkB.

\benchmarkD \ is a set of 2-label pattern unions that are randomly generated. 
The number of items in the Mallows model, $m$, is among $\{20, 30, 40, 50, 60\}$, and $\phi=0.5$. 
The number patterns per union is among $\{2, 3, 4, 5\}$.
The items per label is among $\{3, 5, 7\}$. 
For each combination of the parameters above, 10 random instances are generated.
This benchmark tests the scalability of the two-label solver.

\textbf{Polls} is a synthetic database inspired by the 2016 US presidential election. The data is generated in the way of~\cite{DBLP:conf/sigmod/CohenKPKS18}, with database schema as in Figure~\ref{fig:elections}. The tuples in \textbf{Candidates} and \textbf{Voters}, and the values in each tuple are generated independently. Attributes party and sex have cardinality 2, geographic region cardinality  6, edu and age cardinality  6 (10-year brackets). For age, we assigned values between 20 and 70 in increments of 10, with each value represents a 10-year bracket. We generate 1000 voters falling into 72 demographic groups. For each group, we generate 3 random reference rankings and 3 $\phi$ values $\{ 0.2, 0.5, 0.8 \}$ to construct 9 distinct Mallows models. Each voter is randomly assigned a Mallows  from her group, and a random poll date from two dates, which instantiates the relation \textbf{Polls}. 

\textbf{MovieLens} is a dataset of movie ratings from GroupLens (\url{www.grouplens.org}). In line with previous works~\cite{DBLP:journals/jmlr/LuB14, DBLP:conf/sigmod/CohenKPKS18}, we use the 200 (out of around 3900) most frequently rated movies and ratings from 5980 users who rated at least one of these movies. We learned a mixture of 16 Mallows models using a publicly available tool~\cite{DBLP:conf/webdb/StoyanovichIP16}. We store movie information in a relation $M$(id, title, year, genre).

\textbf{CrowdRank} is a real dataset of movie rankings of 50 Human Intelligence Tasks (HITs) collected on Amazon Mechanical Turk~\cite{DBLP:conf/webdb/StoyanovichJG15}. Each HIT provides 20 movies for 100 workers to rank. Then a mixture of Mallows is mined for each HIT with a publicly-available tools~\cite{DBLP:conf/webdb/StoyanovichIP16}. We selected a HIT with seven Mallows models. CrowdRank also includes worker demographics. We used a publicly available tool~\cite{DBLP:conf/ssdbm/PingSH17} to generate 200,000 synthetic user profiles statistically similar to the original 100 workers, with the Mallows model among the attributes.

\begin{figure*} [t!]
	\centering
	\begin{minipage}[t]{.64\textwidth}
		\scalebox{1}{
			\centering
			\medskip
			\subfloat[3 patterns/union, 3 items/label]{
				\includegraphics[width=0.5\linewidth]{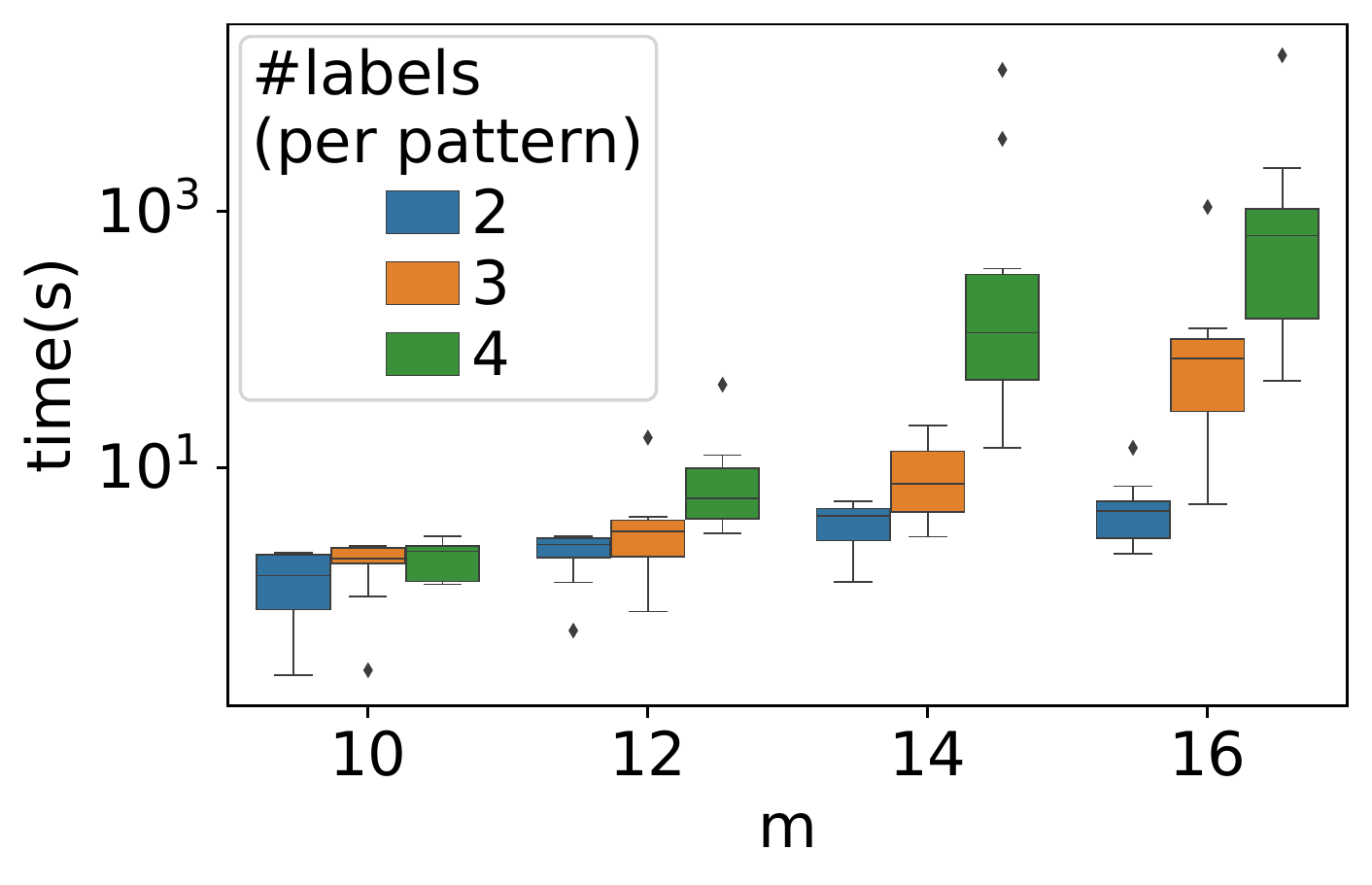}
				\label{fig:scalability_bipartite_1}
			}	\hskip-0.5em
			\subfloat[3 labels/pattern, 3 items/label]{
				\includegraphics[width=0.5\linewidth]{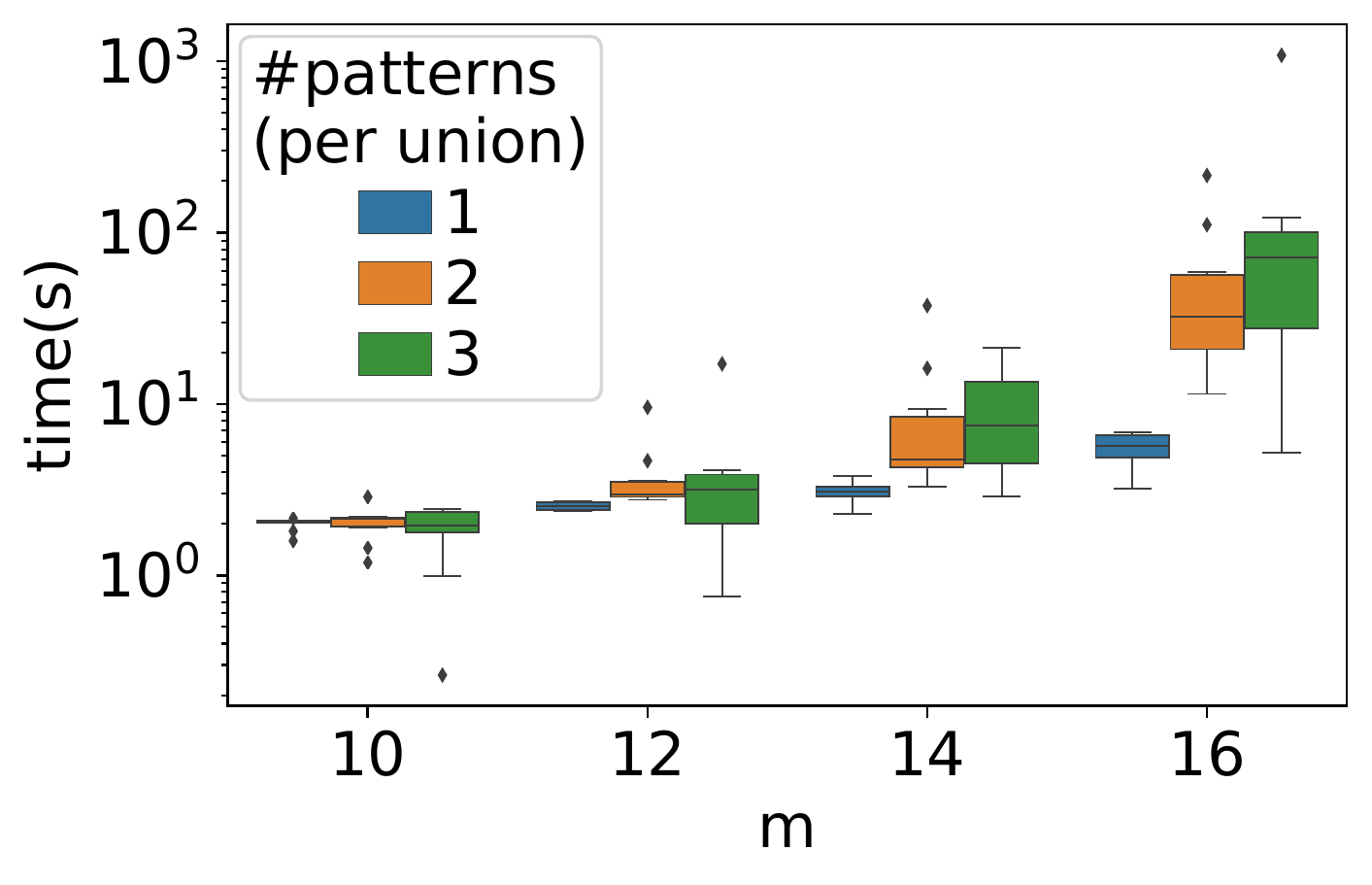}
				\label{fig:scalability_bipartite_2}
			}
		}
		\caption{Scalability of bipartite solver over \benchmarkC}
	\end{minipage}
	\hfill
	\begin{minipage}[t]{.32\textwidth}
		\centering
		\includegraphics[width=\linewidth]{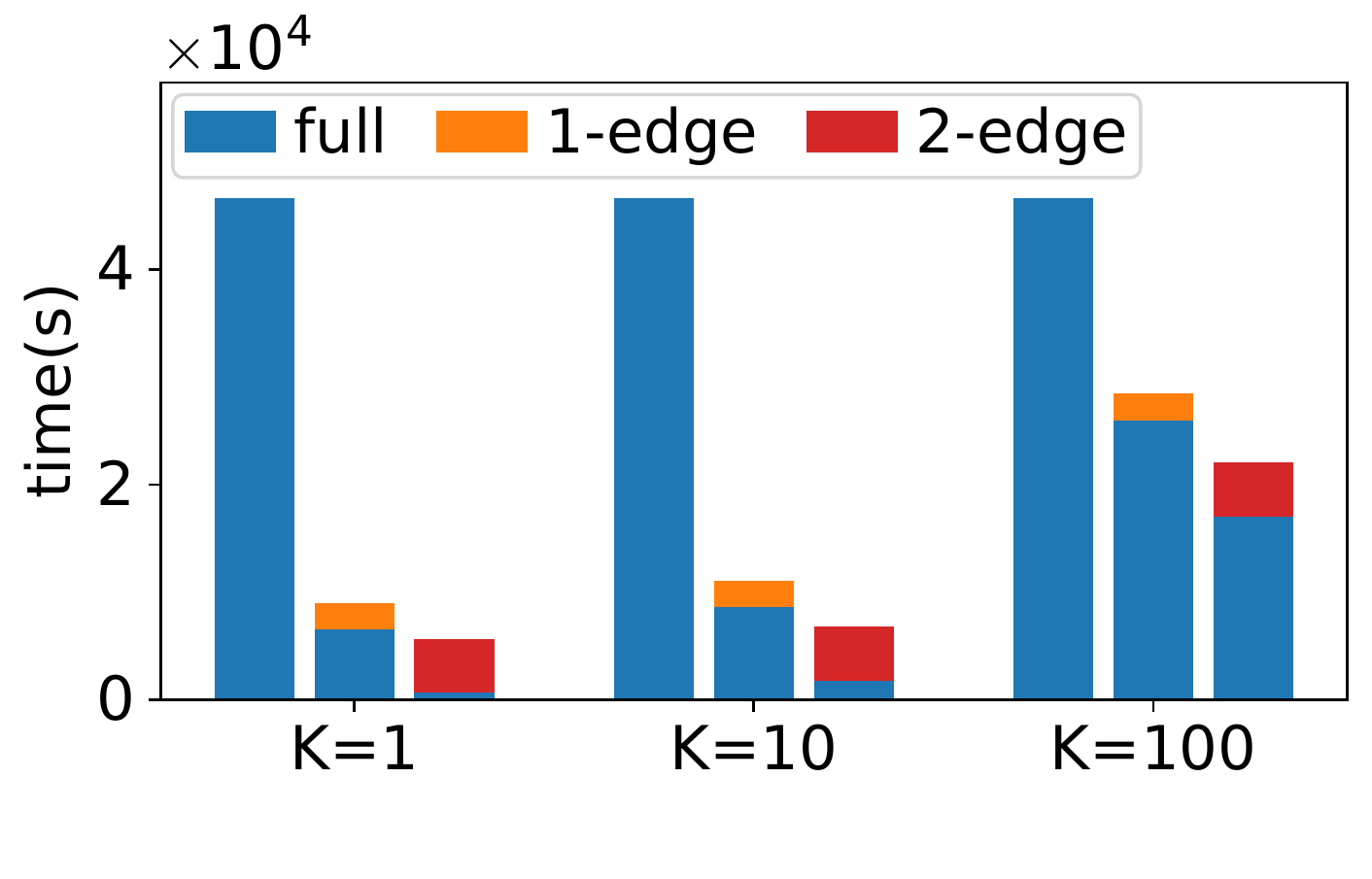}
		\caption{The top-$k$ optimization works well over \textbf{Polls}. The tallest ``full'' bars are baseline. The ``1-edge'' and ``2-edge'' bars first quickly compute upper bounds of all sessions.}
		\label{fig:polls__topK}
	\end{minipage}
	\hfill
\end{figure*}

\begin{figure*} [t!]
	\centering
	\begin{minipage}[t]{.32\textwidth}
		\centering
		\includegraphics[width=\linewidth]{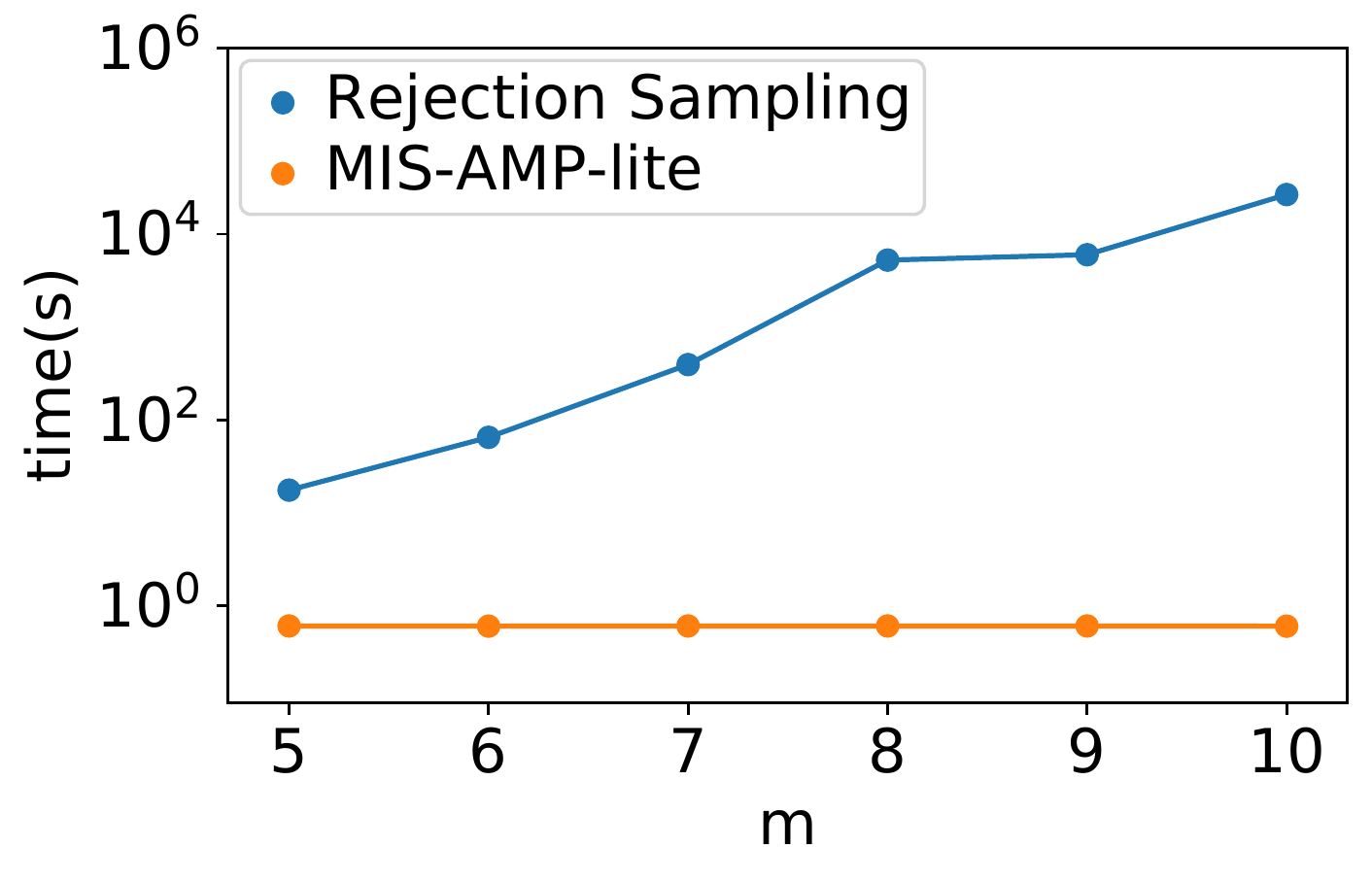}
		\caption{RS does not scale as well as MIS-AMP-lite for query $\sigma_m \succ \sigma_1$ over $\mallows(\langle \sigma_1, \ldots, \sigma_m \rangle, 0.1)$.}
		\label{fig:RS_slow}
	\end{minipage}
	\hfill
	\begin{minipage}[t]{.64\textwidth}
		\scalebox{1}{
			\centering
			\medskip
			\subfloat[\benchmarkA]{
				\includegraphics[width=0.5\linewidth]{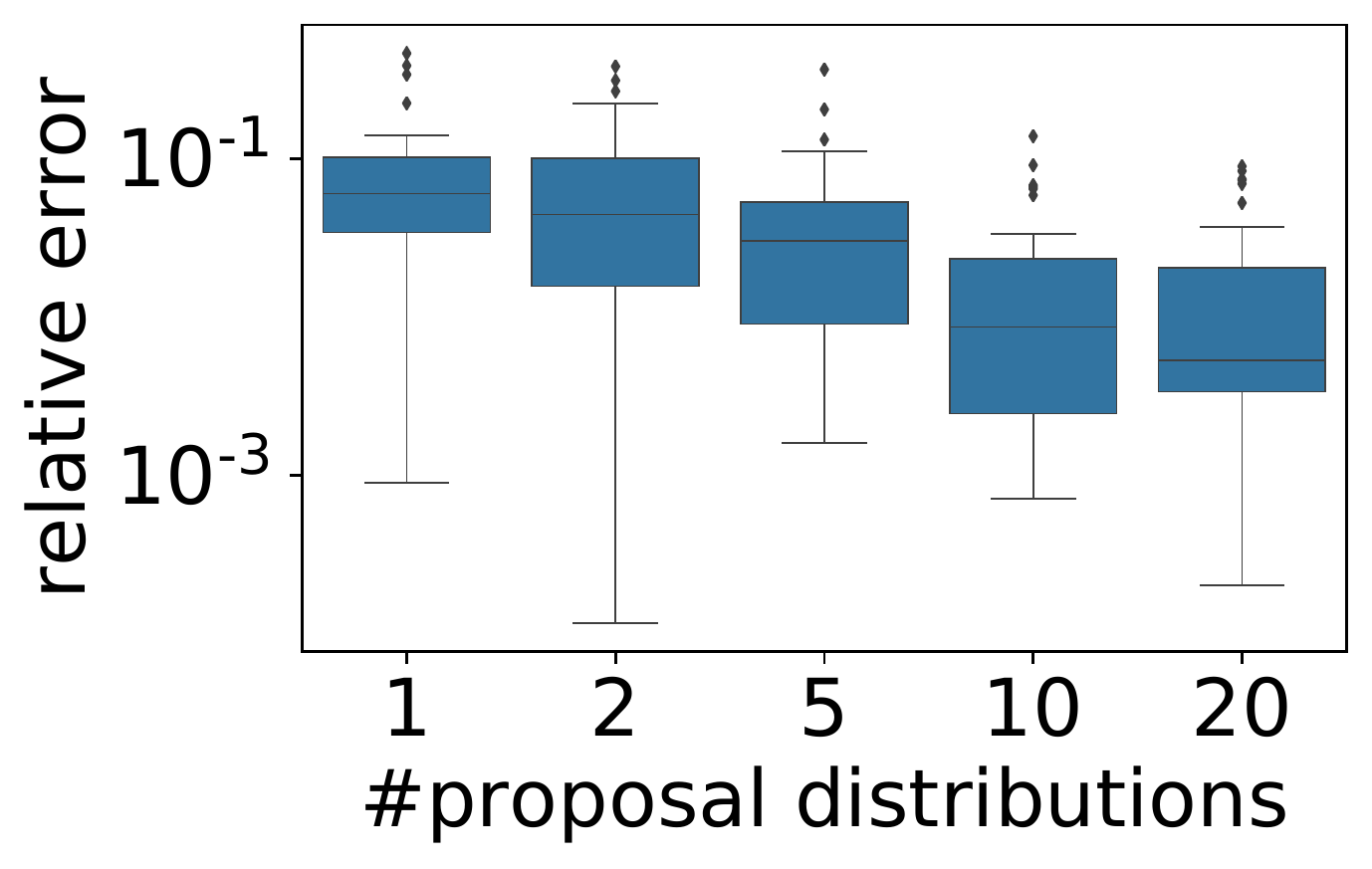}
				\label{fig:MIS-AMP_accuracy_vs_k}
			}
			\subfloat[\benchmarkC, 3 patterns/union, \hfill \break 3 labels/pattern, 3 items/label]{
				\includegraphics[width=0.5\linewidth]{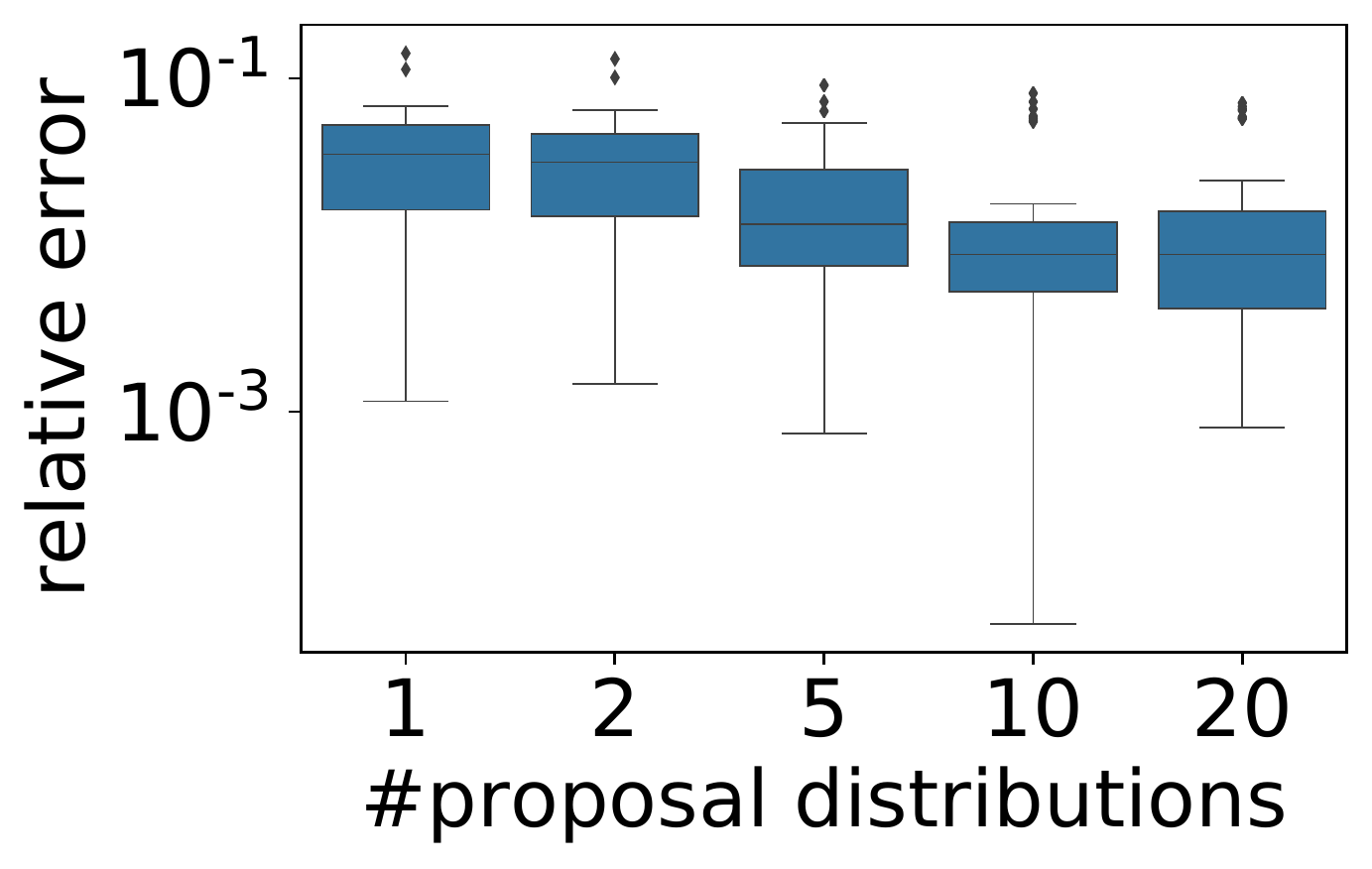}
				\label{fig:MIS-AMP_accuracy_vs_k__benchmarkC}
			}
		}
		\caption{Multi-proposal distributions improve accuracy.}
		\label{fig:all}
	\end{minipage}
	\hfill
\end{figure*}

\subsection{Performance of Exact Solvers}
\label{sec:exp:exact}

In our first experiment, we highlight the relative performance of three exact solvers (Section~\ref{sec:exact}) and the approximate solver MIS-AMP-adaptive (Section~\ref{sec:approx}) over \textbf{Polls} with  20 to 30 candidates, for a Boolean CQ that all solvers can handle:
\[
Q() \leftarrow P(\_,\_;l;r), C(l,p,\val{M},\_,\_,\_), C(r,p,\val{F},\_,\_,\_);
\]
$Q$ asks whether any session prefers a male candidate to a female candidate from the same party $p$.

Figure~\ref{fig:runtime__2label_vs_bipartite_vs_inexclu} compares the running times. Among all solvers, MIS-AMP-adaptive is the most scalable, although, as indicated by the presence of   outliers, the running time of this sampling-based method varies significantly due to randomness.
Among the exact solvers, the two-label solver is faster than the bipartite solver, which is in turn faster than the general solver.
Importantly, MIS-AMP-adaptive is both scalable and accurate for this particular query: 77\% of the instances have relative error under 1\%, and 93\%  have relative error 10\%.  The highest relative error is 63\%.

\medskip
\textbf{General solver over \benchmarkA.}
We evaluate the performance of the general solver over \benchmarkA, where each pattern union $G$ has 3 patterns: $G = g_1 \cup g_2 \cup g_3$. The  solver applies inclusion-exclusion principle to generate pattern conjunctions as follows:
\[
G=g_1 + g_2 + g_3 - (g_1 \wedge g_2) - (g_1 \wedge g_3) - (g_2 \wedge g_3) + (g_1 \wedge g_2 \wedge g_3)
\]
That is, $G$ is decomposed into 7 patterns, and LTM is called to compute the probability for each of them.
Figure~\ref{fig:LTM_runtime_vs_numPatterns} presents the running time of LTM as a function of the number of patterns in a conjunction, showing an exponential increase in running time.

\medskip
\textbf{Two-label solver scalability over \benchmarkD.}
Figure~\ref{fig:scalability_2label} shows the proportions of instances that finished by two-label solver within 10 minutes. 
The two-label solver is sensitive to both total number of items and the number of patterns in a union.
For large pattern unions and large RIM models, the inference algorithm is less likely to finish in time.

\medskip
\textbf{Bipartite solver scalability over \benchmarkC.}
The benchmark has pattern unions of various numbers of patterns per union and various numbers of labels per pattern. 
Recall that the complexity of bipartite solver is $O(m^{q \numPatterns})$ where $m$ is the number of items in RIM model, $q$ is the number of labels per pattern, and $\numPatterns$ is the number of patterns per union. 
The $q \numPatterns$ is the total number of labels in a pattern union, which is a key parameter for bipartite solver.

Figure~\ref{fig:scalability_bipartite_1} shows the running time of bipartite solver with regards to the number of items $m$ and number of labels per pattern, with number of patterns in the union and number of items per label both fixed at 3. The running time increases very fast with both parameters.
Similarly, Figure~\ref{fig:scalability_bipartite_2} shows the running time of bipartite solver with regards to the number of items in RIM model and number of labels per pattern, with number of patterns in union and number of items per label both fixed to be 3. The running time increases very fast with both parameters.  Nonetheless, bipartite solver is practical for lower values of $m$.

\begin{figure*} [t!]
	\centering
	\hfill
	\begin{minipage}[t]{.65\textwidth}
		\centering
		\scalebox{1}{
			\subfloat[typical]{
				\includegraphics[width=0.34\linewidth]{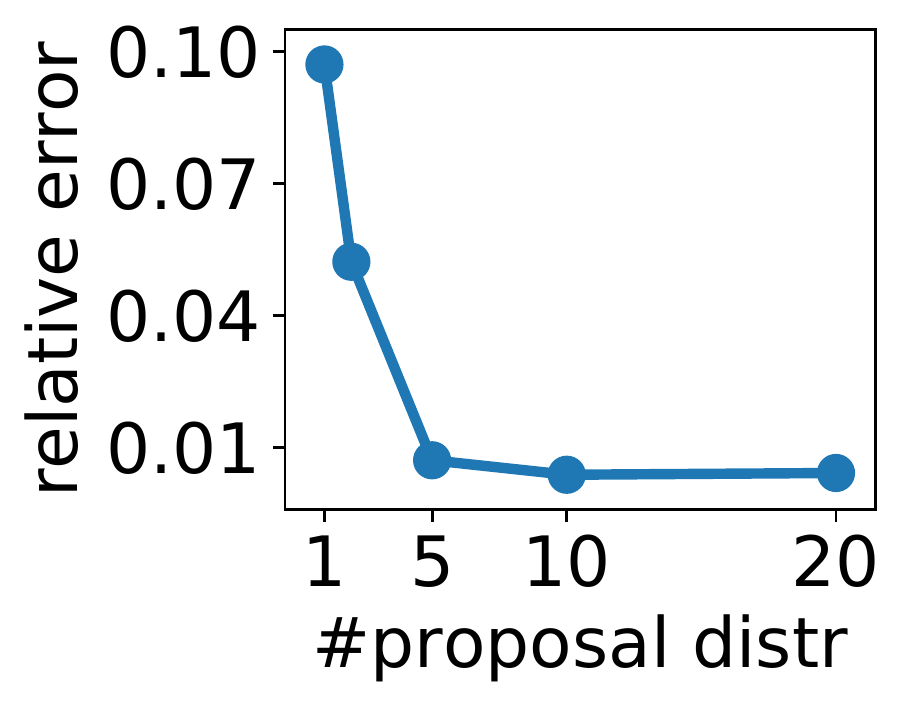}
				\label{fig:MIS-AMP_normal_instance}
			}
			\hspace{-1em}
			\subfloat[atypical]{
				\includegraphics[width=0.315\linewidth]{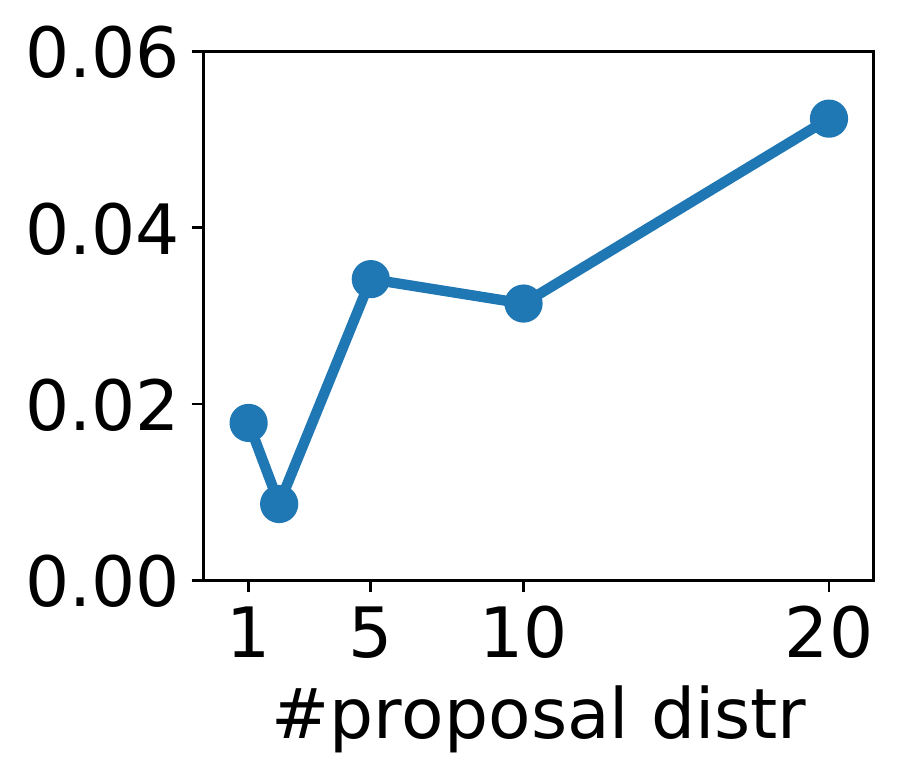}
				\label{fig:MIS-AMP_abnormal_instance}
			}
			\hspace{-1em}
			\subfloat[no comp. for \textbf{(b)}]{
				\includegraphics[width=0.31\linewidth]{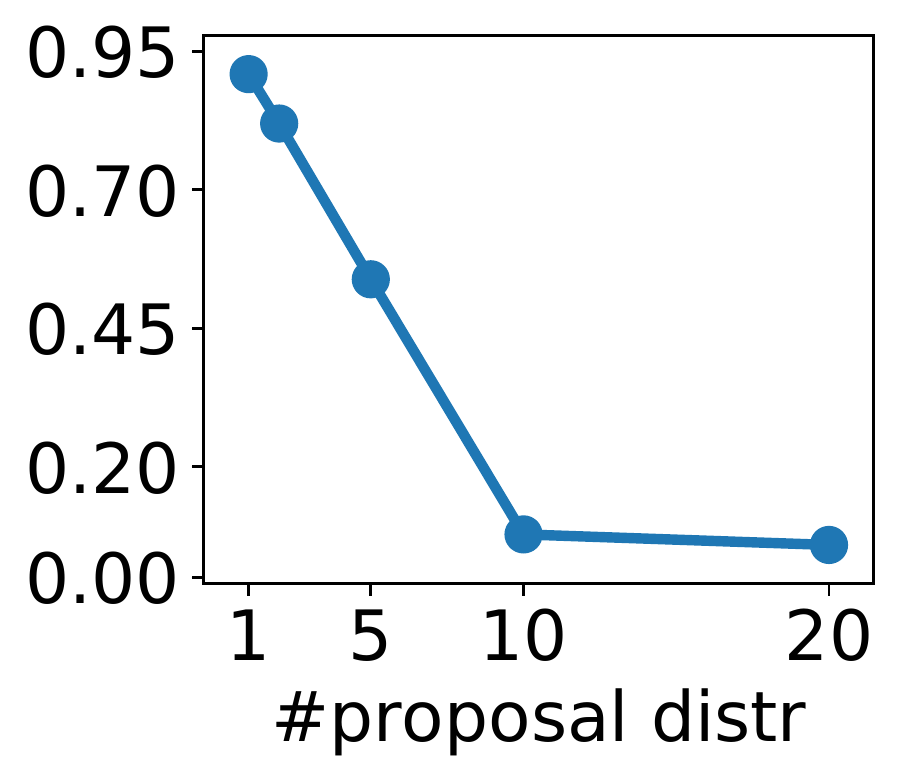}
				\label{fig:MIS-AMP_abnormal_instance_without_compensation}
		}}
		\caption{MIS-AMP-lite over \benchmarkA : (a) more proposal distributions improve accuracy; (b) an atypical instance; (c) accuracy improves with more proposal distributions again, after turning off compensation for (b).}
		\label{fig:MIS-AMP-lite_instances}	
	\end{minipage}
	\hfill
	\begin{minipage}[t]{.31\textwidth}
		\centering
		\includegraphics[width=.75\linewidth]{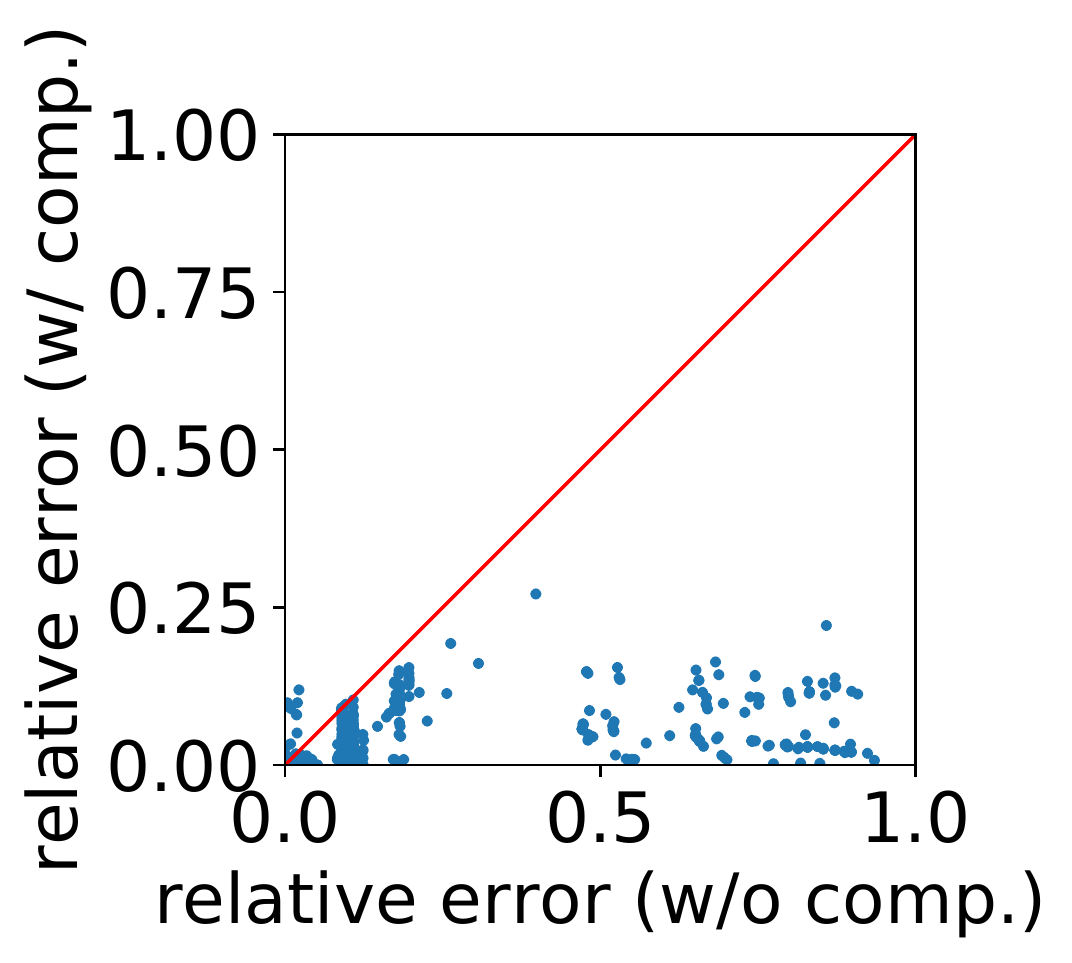}
		\caption{Compensation of MIS-AMP-lite improves the accuracy of estimation on {\bf Benchmark-C}.}
		\label{fig:MIS-AMP_compensation}
	\end{minipage}
    \hfill
\end{figure*}

\medskip
\textbf{Top-$k$ optimization over Polls.}
Next, we evaluate the performance of the \textbf{top-$k$ optimization} on \textbf{Polls} with 16 candidates. The query is the following.  (Note that it contains a self-join.)
\begin{equation*}
\begin{split}
Q() \leftarrow & P(\_,date;c_1;c_2), P(\_,date;c_1;c_3), P(\_,date;c_1;c_4), \\
& C(c_1,p,\_,\_,\_,\val{NE}), C(c_2,p,\_,\_,\_,\val{MW}), date=\val{"5/5"}, \\
& C(c_3,\_,\_,age,\_,\val{NE}), C(c_4,\_,\val{M},\_,\val{BA},\_), age = 50;
\end{split}
\end{equation*}

Figure~\ref{fig:polls__topK} displays the running times of evaluating this query under $k=[1, 10, 100]$. Three tallest bars represent the simple strategy of calculating all sessions. The lower bars with 2 colors represent top-$k$ optimization. The ``1-edge'' label means calculating upper bounds of all sessions by selecting only one edge from each pattern. The ``full'' label below ``1-edge'' is the amount of time spent on evaluating exact probabilities of sessions in descending order of their upper bounds until there are $k$ sessions having probabilities higher the probabilities or upper bounds of rest sessions. The ``2-edge'' label means selecting 2 edges for more accurate upper bounds. As a result, the ``full'' label below ``2-edge'' means fewer sessions to calculate. In Figure~\ref{fig:polls__topK}, applying ``1-edge'' and ``2-edge'' speeds up the evaluation of $k=1$ by 5.2 times and 8.2 times, respectively. Even for $k=100$, the speedup of applying ``1-edge'' and ``2-edge'' reaches 1.6 and 2.1, respectively.

\medskip
{\bf In summary}, all exact solvers have exponential complexity with query size. The two-label solver is the fastest, while the bipartite solver is also efficient and can be used also for two-label queries as a special case. These two solvers are also effective in scope of the top-$k$ optimization.

\begin{figure*} [t!]
	\centering
	\begin{minipage}[t]{.64\textwidth}
		\scalebox{1}{
			\centering
			\medskip
			\subfloat[Overhead time]{
				\includegraphics[width=0.49\linewidth]{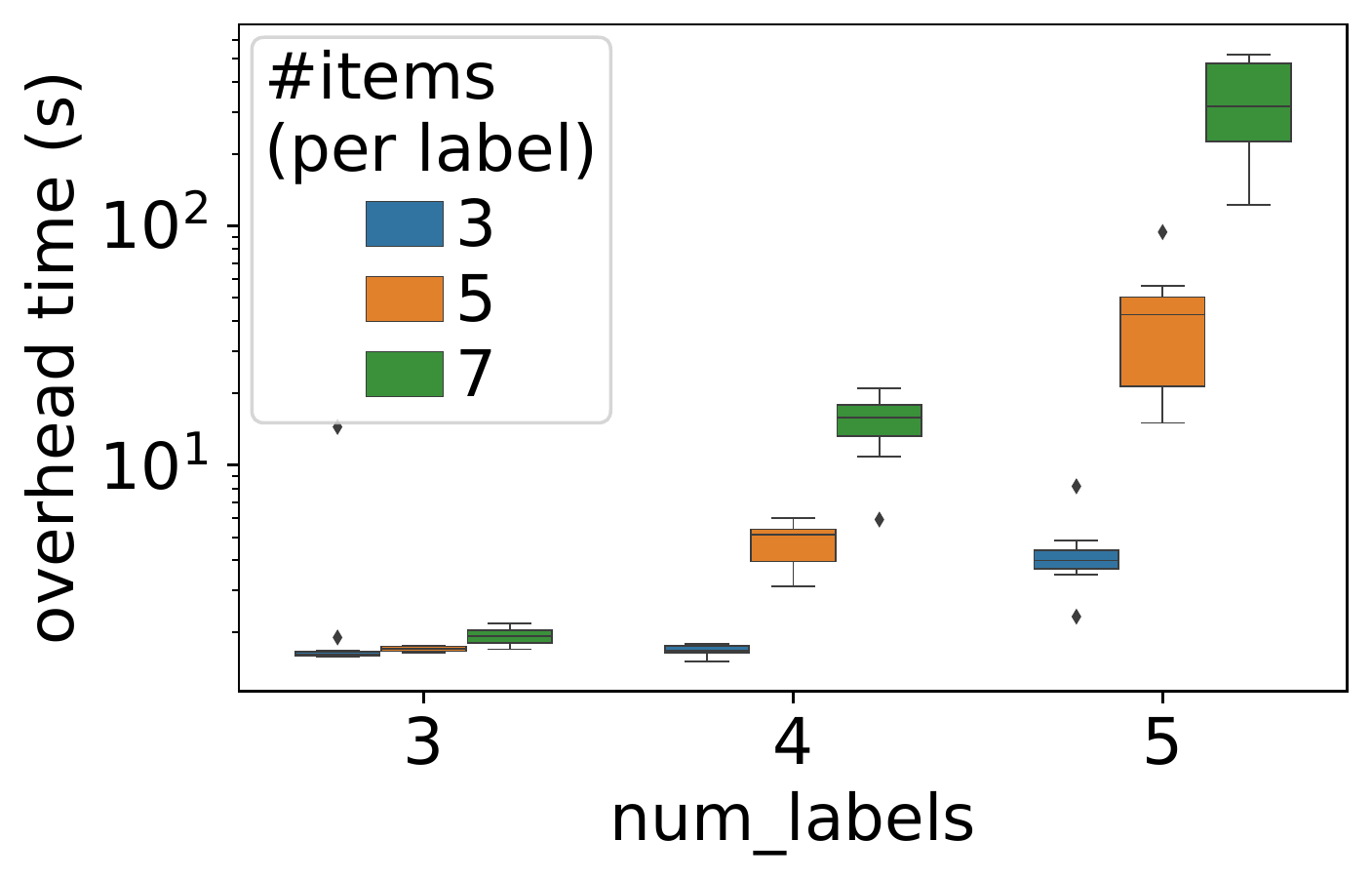}
				\label{fig:MIS-AMP-adaptive_overhead}
			}
			\hfill
			\subfloat[Convergence time]{
				\includegraphics[width=0.49\linewidth]{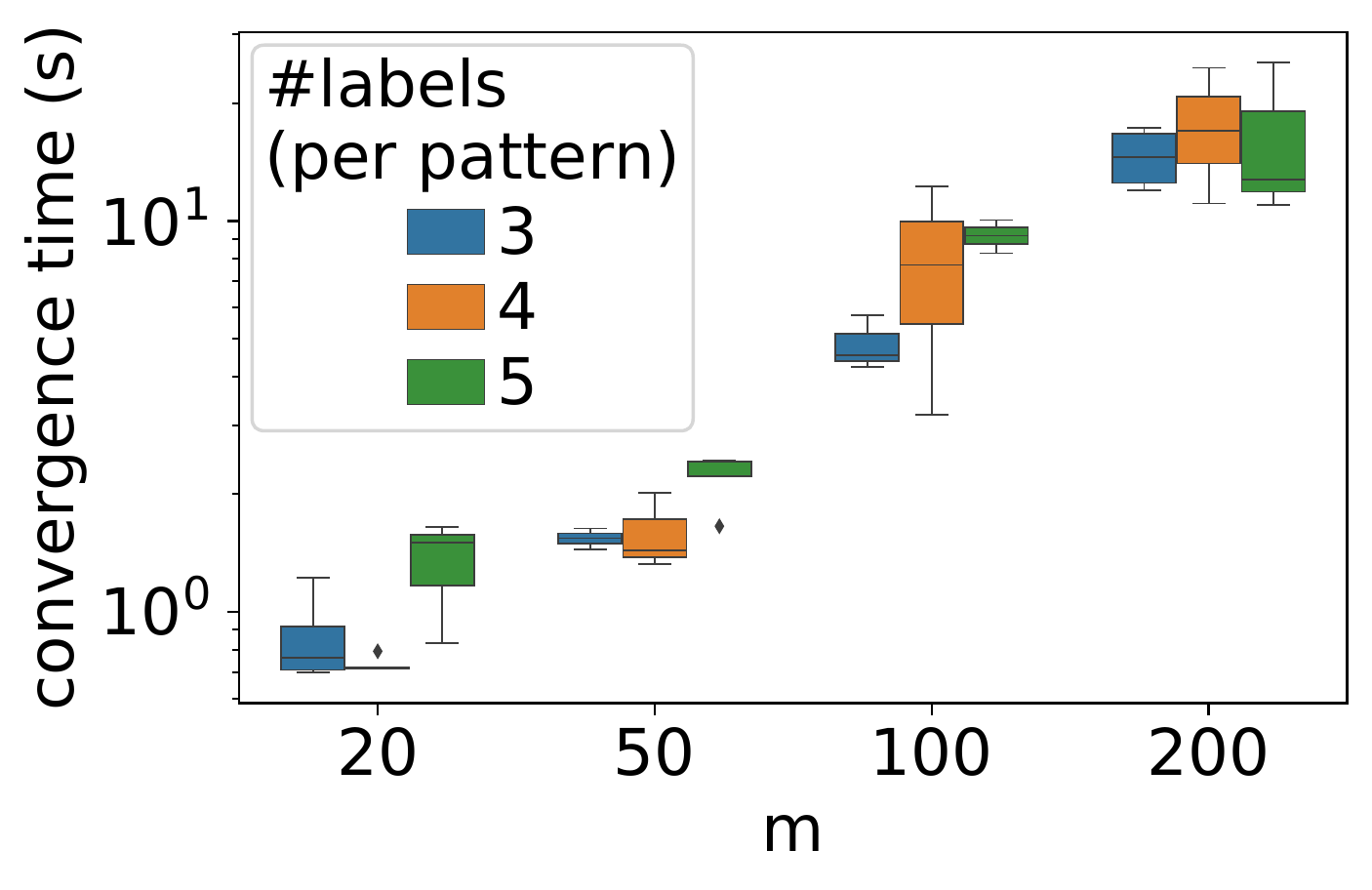}
				\label{fig:MIS-AMP-adaptive_convergence}
			}
		}
		\caption{Scalability of MIS-AMP-adaptive over {\bf Benchmark-B}.}
	\end{minipage}
	\hfill
	\begin{minipage}[t]{.32\textwidth}
    \centering
	\includegraphics[width=\linewidth]{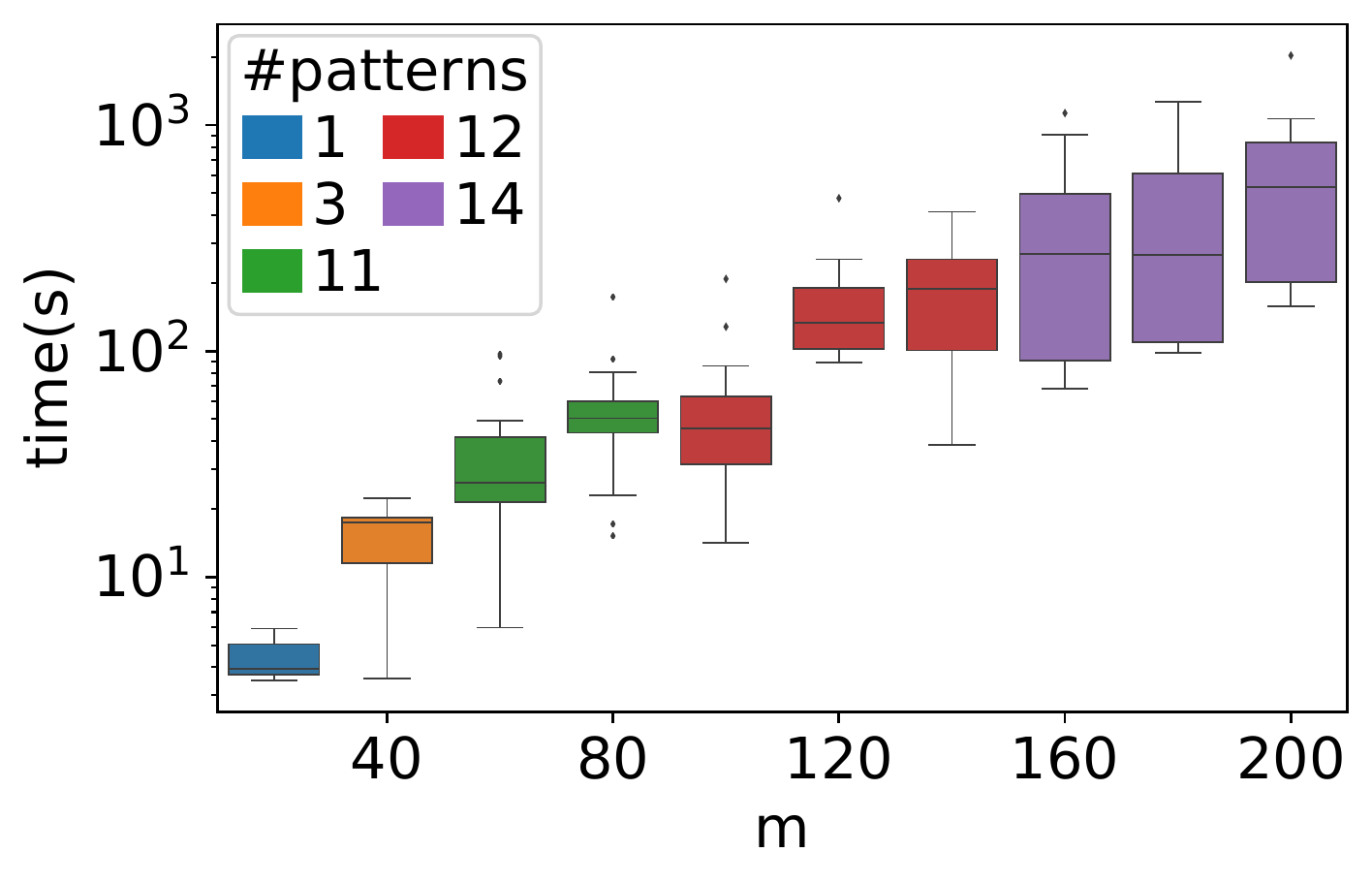}
	\caption{MIS-AMP-adaptive runtime over \textbf{MovieLens}.}
	\label{fig:MIS-AMP-adaptive__movielens__runtime}
	\end{minipage}
	\hfill
\end{figure*}

\subsection{Performance of Approximate Solvers}
\label{sec:exp:approx}

\medskip
\textbf{Rejection Sampling is inefficient for rare events.}
We constructed a simple low-probability query $\sigma_m \succ \sigma_1$ for $\mallows(\bsigma, 0.1)$, where $\bsigma = \ranking{\lst{\sigma}}$.
When increasing $m$, the $\Pr(\sigma_m \succ \sigma_1 | \bsigma, 0.1)$ decreases exponentially, and RS needs EXP($m$) samples for convergence.
In this experiment, we generate 6 Mallows models with $m\in \{5,6,7,8,9,10\}$.
For each Mallows, we run RS and MIS-AMP-lite 10 times. The exact values of $\Pr(\sigma_m \succ \sigma_1 | \bsigma, 0.1)$ are pre-calculated.
RS stops running when the estimated probability is within 1\% relative error.
(Note that this is an optimistic stopping condition for RS, since the algorithm would not yet be able to determine that it converged.)
MIS-AMP-lite is set to have only 1 proposal distribution.
Figure~\ref{fig:RS_slow} shows that RS running time increases exponentially with $m$, while MIS-AMP-lite is much more scalable.

\medskip
\textbf{MIS-AMP-lite over \benchmarkA, \benchmarkC.} The number of proposal distributions is a critical parameter for MIS-AMP-lite. In this experiment, MIS-AMP-lite is executed with 1, 2, 5, 10, 20 proposal distributions.

Figure~\ref{fig:all} gives the distributions of relative errors of MIS-AMP-lite as a function of the number of proposal distributions on \benchmarkA \ and \benchmarkC \ with the number of patterns in union, number of labels per pattern, and number of items per label fixed to be 3. Accuracy improves as the number of proposal distributions increases, and plateaus at around 20 distributions.  Overall, MIS-AMP-lite shows low relative error.

Figure~\ref{fig:MIS-AMP_normal_instance} complements these cumulative results, showing accuracy of MIS-AMP-lite on a specific  instance, where 10 distributions is a good choice.  Further, we investigated an atypical instance in Figure~\ref{fig:MIS-AMP_abnormal_instance}. Its relative error was reduced mainly by the compensation, and adding proposal distributions kept increasing accuracy after turning off compensation, as shown in Figure~\ref{fig:MIS-AMP_abnormal_instance_without_compensation}.

\medskip
\textbf{MIS-AMP-lite over \benchmarkC.}
To test the effectiveness of compensation systematically, we ran MIS-AMP-lite with one proposal distribution over \benchmarkC.
Figure~\ref{fig:MIS-AMP_compensation} shows that the accuracy of most instances improved by compensation  (blue dots under the red line), especially those near the lower right corner, corresponding to instances where relative error was very high (close to 100\%) before compensaiton, and was reduced dramatically by applying compensation.

\medskip
\textbf{MIS-AMP-adaptive over \benchmarkB.} MIS-AMP-adaptive has two stages, proposal distribution construction and sampling. 
Figure~\ref{fig:MIS-AMP-adaptive_overhead} shows the overhead due to proposal distribution construction, fixing 100 items in Mallows model and 3 patterns in union. As expected, the overhead increases sharply with the number of labels, especially when there are many items per label.
But once proposal distributions are constructed, sampling converges quickly. 

Figure~\ref{fig:MIS-AMP-adaptive_convergence} shows the sampling time, fixing 2 patterns in union and 5 items per label. The sampling time increases only moderately with the number of items in Mallows model, and the query size (number of labels) doesn't have significant impact on sampling time. Note that due to the randomness of sampling procedure, here we repeated the sampling 3 times and select the median value to plot in figure.

\medskip
\textbf{MIS-AMP-adaptive over MovieLens.} We vary the number of movies $m$ from 40 to 200 to test scalability with:
\begin{equation*}
\begin{split}
Q() \leftarrow & P(\_; \val{223}; \val{111}), P(\_; x; \val{111}), P(\_; x; y),\\
& M(x, \_, year_1, genre), year_1 >= \val{1990},\\
& M(y, \_, year_2, genre), year_2 < \val{1990};
\end{split}
\end{equation*}
The query asks whether the movie \textit{Clerks} (id 223) is preferred to \textit{Taxi Driver} (id 111), and whether some movie released after 1990 is preferred to a movie before 1990 and also to \textit{Taxi Driver}.
Figure~\ref{fig:MIS-AMP-adaptive__movielens__runtime} shows the running time of MIS-AMP-adaptive over the sessions. Note that when number of movies $m$ increases, there are more genres in the dataset, yielding more patterns in the pattern union.

\medskip
{\bf In summary}, the approximate solvers are scalable and accurate. Multiple proposal distributions help them reach the important regions of the target distribution.
Although MIS-AMP-lite prunes many modals, the compensation step works. When applying MIS-AMP solvers to large dataset such as \textbf{MovieLens}, the overhead of proposal distribution construction is significant. But once the proposal distributions are ready, MIS-AMP solvers converge fast.

\begin{figure}[t!]
	\centering
	\includegraphics[width=.7\linewidth]{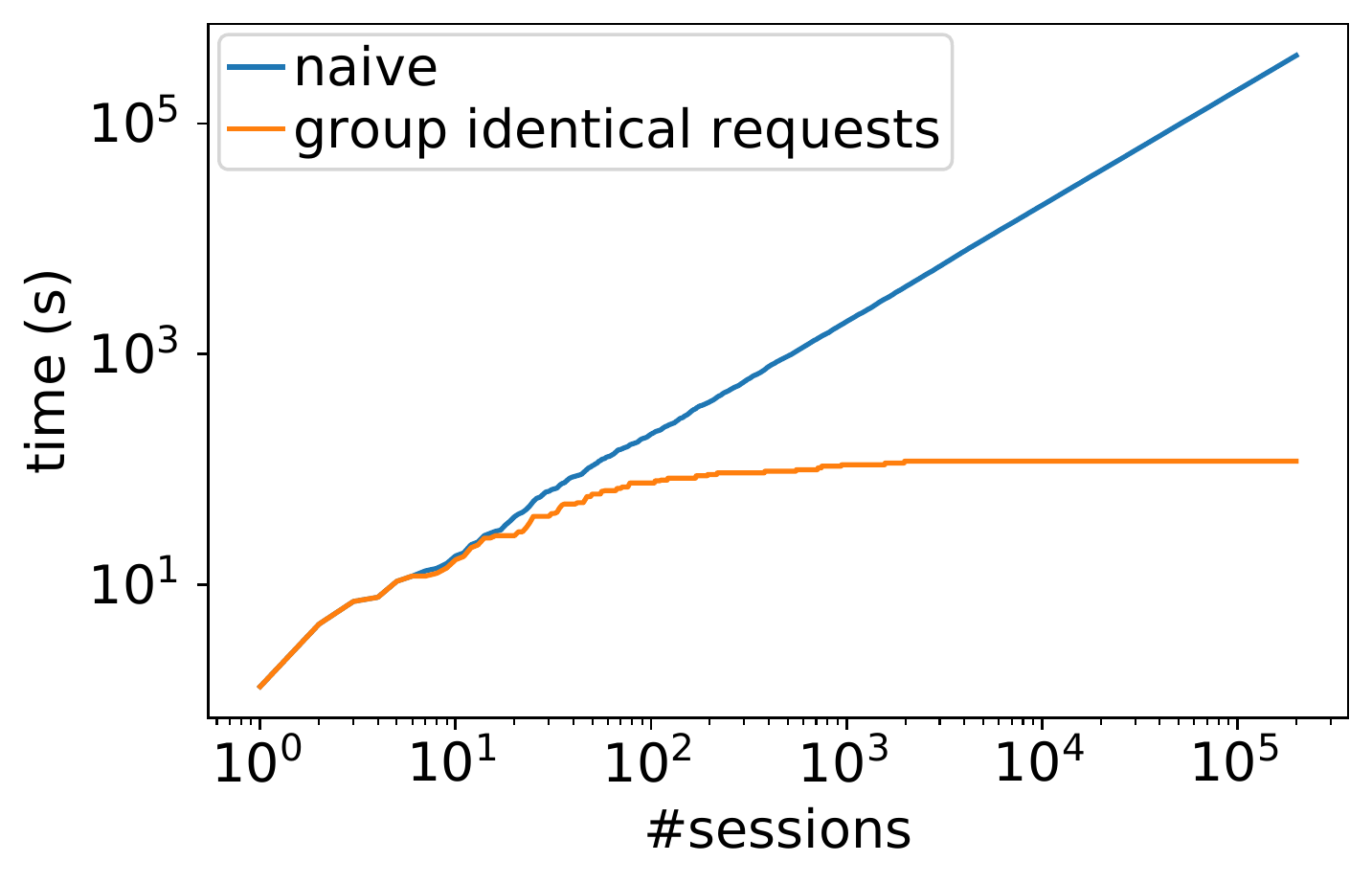}
	\caption{Scalability over 200K sessions in \textbf{CrowdRank}.}
	\label{fig:crowdrank}
\end{figure}

\subsection{Scalability over Sessions}

When evaluating a query, multiple sessions may share the same RIM model and pattern union.
RIM-PPD groups identical requests before invoking inference solvers, realizing performance gains.
We illustrate scalability in the number of sessions using a query that  asks whether a user prefers a movie with the leading actor of their gender to a movie with the leading actor around their age.
Focus on short ($<$ 90 min) movies that are preferred to some Thriller. 
\begin{equation*}
\begin{split}
Q() \leftarrow & P(v; m_1; m_2), P(v; m_2; m_3), V(v, sex, age), \\
& M(m_1, \_, sex, \_, \val{short}), M(m_2, \_, \_, age, \val{short}), \\
& M(m_3, \val{Thriller}, \_, \_);
\end{split}
\end{equation*}

Figure~\ref{fig:crowdrank} shows the results of running the general solver over \textbf{CrowdRank} with 200,000 sessions. The naive implementation runs in linear time in the number of sessions, while grouping requests quickly converged after 118 seconds.
\section{Conclusions}
\label{sec:conclusions}

In this work, we developed methods for answering computationally hard queries over probabilistic preferences, where we enable users to express preferences over item attributes in the form of values or variables. 
To evaluate this class of hard queries, we developed a general solver that applies inclusion-exclusion principle.
Then, we took the optimization opportunities in two-label patterns and bipartite patterns, significantly reducing query evaluation time. 
Scalability was further improved by approximate solvers, where we studied the posterior distributions of pattern unions over Mallows models, and applied Multiple Importance Sampling to effectively estimate the Mallows posterior probability.

Future directions include supporting additional aggregation queries (e.g., average age of voters who prefer a republican to a democrat), and incorporating probabilistic  preference models beyond RIM~\cite{Fligner1986, DBLP:conf/nips/LebanonM07}.

\balance

\bibliographystyle{abbrv}
\bibliography{references.bib}

\end{document}